\documentclass[preliminary]{eptcs}
\providecommand{\event}{GCM 2021}

\usepackage[utf8]{inputenc}
\usepackage[english]{babel}

\usepackage[titletoc]{appendix}
\usepackage{stmaryrd}
\usepackage{amssymb}
\usepackage{amsmath}
\usepackage{graphicx}
\usepackage{enumerate}
\usepackage{enumitem}
\usepackage{multicol}
\usepackage{subcaption}
\usepackage{proof}
\usepackage{amsthm}
\usepackage[normalem]{ulem}
\usepackage{array}
\usepackage{lipsum}

\theoremstyle{plain}

\newtheorem{theorem}{Theorem}
\newtheorem{proposition}{Proposition}

\theoremstyle{definition}

\newtheorem{definition}{Definition}

\theoremstyle{remark}
\newtheorem{example}{Example}
\newtheorem*{remark}{Remark}

\usepackage{tikz}
\usetikzlibrary{patterns,arrows, topaths, calc, positioning}
\tikzset{>=stealth}
\tikzstyle{node} = [circle, minimum size = 1.4mm, inner sep = 0mm, color=black, fill]
\tikzstyle{hyperedge} = [rectangle, minimum width = 5mm, minimum height = 5mm, draw, inner sep = 0mm]
\tikzstyle{HG} = [align = center]
\tikzstyle{circledge} = [circle, minimum size = 7mm, inner sep = 0mm, color=black, draw]

\pgfdeclarepatternformonly{NELines}{\pgfqpoint{-1pt}{-1pt}}{\pgfqpoint{4pt}{4pt}}{\pgfqpoint{3pt}{3pt}}%
{
	\pgfsetlinewidth{0.4pt}
	\pgfpathmoveto{\pgfqpoint{0pt}{0pt}}
	\pgfpathlineto{\pgfqpoint{1pt}{1pt}}
	\pgfpathmoveto{\pgfqpoint{2pt}{2pt}}
	\pgfpathlineto{\pgfqpoint{3.1pt}{3.1pt}}
	\pgfusepath{stroke}
}

\newcommand{\eqdef}{\mathrel{\mathop:}=}
\newcommand{\SG}{\mathrm{SG}}
\mathchardef\mhyp="2D

\newcommand{\ts}{\mathop{\mathrm{ts}}}
\newcommand{\corr}{\mathrel{\triangleright}}
\newcommand{\FB}{\mathop{\mathfrak{F}}}
\title{Grammars Based on a Logic of Hypergraph Languages}
\author{Tikhon Pshenitsyn
	\institute{Department of Mathematical Logic and Theory of Algorithms\\Faculty of Mathematics and Mechanics\\
		Lomonosov Moscow State University
		\\GSP-1, Leninskie Gory, Moscow, 119991, Russian Federation}\thanks{The study was supported by RFBR, project number 20-01-00670, by the Theoretical Physics and Mathematics Advancement Foundation ``BASIS'', and by the Interdisciplinary Scientific and Educational School of Moscow University ``Brain, Cognitive Systems, Artificial Intelligence''.}
	\email{ptihon at yandex.ru}
}
\def\authorrunning{T. Pshenitsyn}
\def\titlerunning{Grammars Based on a Logic of Hypergraph Languages}

\begin{document}
	\maketitle
\begin{abstract}
	The hyperedge replacement grammar (HRG) formalism is a natural and well-known generalization of context-free grammars. HRGs inherit a number of properties of context-free grammars, e.g. the pumping lemma. This lemma turns out to be a strong restriction in the hypergraph case: it implies that languages of unbounded connectivity cannot be generated by HRGs. We introduce a formalism that turns out to be more powerful than HRGs while having the same algorithmic complexity (NP-complete). Namely, we introduce hypergraph Lambek grammars; they are based on the hypergraph Lambek calculus, which may be considered as a logic of hypergraph languages. We explain the underlying principles of hypergraph Lambek grammars, establish their basic properties, and show some languages of unbounded connectivity that can be generated by them (e.g. the language of all graphs, the language of all bipartite graphs, the language of all regular graphs).
\end{abstract}

\section{Introduction: Productions vs Types}\label{sec_intr_cfg_lg}
	Formal grammar theory is an area at the intersection of linguistics, logic, programming etc., which studies an issue of how complex, unboundedly large, but in some sense regular families of objects (strings, terms, graphs, pictures, ...), can be described using some class of finite-sized mechanisms. There exist numerous kinds of formal grammars based on different principles. For instance, standard context-free grammars deal with strings; they consist of terminal and nonterminal alphabets, a nonterminal start symbol $S$, and the list of productions of the form $X\to\alpha$, which allow one to replace the nonterminal symbol $X$ by the string $\alpha$. 
	\begin{example}
		Look at the following ``linguistic'' example (which is extremely primitive from the linguistic point of view):
		\begin{center}
			\begin{tabular}{rl}
				Productions: & $
				S\to \mathit{NP} \mbox{ sleeps} \qquad \mathit{NP} \to \mbox{the } N \qquad N \to \mbox{cat}
				$
				\\
				Derivation: & $
				S \Rightarrow \mathit{NP} \mbox{ sleeps} \Rightarrow \mbox{the } N \mbox{ sleeps} \Rightarrow \mbox{the cat sleeps}
				$
			\end{tabular}
		\end{center}
	\end{example}
	In the production $[S\to \mathit{NP} \mbox{ sleeps}]$, there is exactly one terminal object (\emph{sleeps} here) in its right-hand side. Let us transform this production as follows: $\mbox{sleeps} \corr \mathit{NP}\backslash S$. The $\corr$ sign is to be read as ``is of the type'', and $\mathit{NP}\backslash S$ stands for the type of such objects $u$ that, whenever we add an object $v$ of the type $\mathit{NP}$ ($\mathit{NP}$ stands for the class of \emph{singular noun phrases}, which includes e.g. phrases \textit{the cat}, \textit{Helen}, \textit{a colorless green idea} etc.) to the left of $u$, $vu$ forms an object of the type $S$ ($S$ stands for \emph{sentence}). Therefore, $\mbox{sleeps}\corr \mathit{NP}\backslash S$ means that the verb \emph{sleeps} is such an object that, whenever a noun phrase (singular) appears to its left, they together form a sentence. This is correct, to a first approximation: \textit{the cat sleeps}, \textit{Helen sleeps}, \textit{a green colorless idea sleeps} etc. are correct English sentences. Similarly, we can transform the production $[\mathit{NP} \to \mbox{the } N]$ into a correspondence: $\mbox{the}\corr \mathit{NP}/N$ ($N$ stands for singular nouns with their dependents that do not represent a specific object but rather a class of objects: \textit{cat}, \textit{colorless green idea}). Note that the direction of the division is different from the previous one. The following reduction laws hold: the left one $A, A\backslash B \to B$, and the right one $B/A, A\to B$. The type $A/B$ ($B\backslash A$) can be understood as the type of functions that take an argument of the type $B$ from the right (from the left resp.) and return a value of the type $A$.
	
	In linguistics, however, it is not enough to have the above reduction laws to describe language phenomena of interest. E.g. sometimes it is useful to have the rule $\mathit{NP}\to S/(\mathit{NP}\backslash S)$ or the rule $A/B,B/C\to A/C$. Besides, there is a need for an operation that would store pairs of units of certain types. E.g. when we consider sentences like \textit{Tim gave the lemon to Melany and the lime to Amelie}, we would like to assign the type $\left((\mathit{NP}\cdot \mathit{PP})\backslash (\mathit{NP}\cdot \mathit{PP})\right)/(\mathit{NP}\cdot \mathit{PP})$ to the word \textit{and} (where $\mathit{PP}$ is the type of prepositional phrases with \textit{to}: \textit{to Melany}, \textit{to Amelie}) since it receives a pair consisting of a noun phrase and a prepositional phrase from the right, and a similar pair from the left. $A\cdot B$ is the type of pairs $uv$ such that $u$ is of the type $A$, and $v$ is of the type $B$; therefore, $A \cdot B$ works as pairwise concatenation.
	
	The above connectives $\backslash,\cdot,/$ belong to the Lambek calculus ($\mathrm{L}$) --- a logical calculus introduced in \cite{Lambek58}. Types in the Lambek calculus are built from primitive types $p_1,p_2,\dotsc\in Pr$ using $\backslash,\cdot,/$. We focus on the Lambek calculus in the Gentzen style; this means that it deals with sequents, which are structures of the form $A_1,\dots,A_n\to A$ for $A_i$, $A$ being types ($A_1,\dotsc,A_n$ is called an antecedent, and $A$ is called a succeedent). The calculus $\mathrm{L}$ has one axiom and six inference rules:
	
	\begin{center}
		\begin{tabular}{cccc}
		$$\infer[(\backslash\to)]{\Gamma, \Pi, A \backslash B, \Delta \to C}{\Pi \to A & \Gamma, B, \Delta \to C}
		$$
		&
		$$\infer[(\to\backslash)]{\Pi \to A \backslash B}{A, \Pi \to B}
		$$
		&
		$$\infer[(\cdot\to)]{\Gamma, A \cdot B, \Delta \to C}{\Gamma, A, B, \Delta \to C}
		$$
		&
		$$
		\infer[(\mathrm{Ax})]{A\to A}{} 
		$$
		\\
		&&&\\
		$$
		\infer[(/\to)]{\Gamma, B / A, \Pi, \Delta \to C}{\Pi \to A & \Gamma, B, \Delta \to C}
		$$
		&
		$$\infer[(\to/)]{\Pi \to B / A}{\Pi, A \to B}
		$$
		&
		$$\infer[(\to\cdot)]{\Pi, \Psi \to A \cdot B}{\Pi \to A & \Psi \to B}
		$$
		&
		\\
		\end{tabular}
	\end{center}
	An axiom can be considered as an inference rule with zero premises. Hereinafter, primitive types are denoted by small Latin letters (in particular, from now on we write $s$ instead of $S$, $\mathit{np}$ instead of $\mathit{NP}$ etc.), types are denoted by capital Latin letters, and sequences of types are denoted by capital Greek letters; besides, $\Pi,\Psi$ above are nonempty. A sequent $\Pi\to A$ is called \textit{derivable in $\mathrm{L}$} (denoted $\mathrm{L}\vdash \Pi\to A$) if it can be obtained by applications of inference rules. 
	\begin{example}
		The following sequents are derivable in $\mathrm{L}$ (their derivations are presented below them):
		\vspace{-0.3cm}
			\begin{multicols}{3}
				\begin{itemize}
					\item $\mathit{np}/n, n, \mathit{np}\backslash s\to s$;
					\item $\mathit{np}\to s/(\mathit{np}\backslash s)$;
					\item $p\to (p\cdot q)/q$.
				\end{itemize}
			\end{multicols}
		\vspace{-0.55cm}
		$$
		\infer[(\backslash \to)]
		{
			\mathit{np}/n, n, \mathit{np}\backslash s\to s
		}
		{
			s\to s & 
			\infer[(/\to)]
			{
				\mathit{np}/n, n\to \mathit{np}
			}
			{
				\mathit{np}\to \mathit{np} & n\to n
			}
		}
		\qquad
		\infer[(\to /)]
		{
			\mathit{np}\to s/(\mathit{np}\backslash s)
		}
		{
			\infer[(\backslash \to)]
			{
				\mathit{np},\mathit{np}\backslash s\to s
			}
			{
				s\to s & \mathit{np}\to \mathit{np}
			}
		}
		\qquad
		\infer[(\to /)]
		{
			p\to (p\cdot q)/q
		}
		{
			\infer[(\to\cdot)]
			{
				p, q\to p\cdot q
			}
			{
				p\to p & q\to q
			}
		}
		$$
	\end{example}
	Finally, the above reasonings bring us to using the Lambek calculus as a grammar formalism. A grammar is a correspondence (i.e., a binary relation) $\triangleright$ between terminal objects and types along with some fixed type $S$; a string is said to be correct if its elements can be replaced by coresponding types in such a way that the sequent composed of these types in the antecedent and of $S$ in the succeedent is derivable in $\mathrm{L}$.
	
	Another feature characterizing the Lambek calculus is its language models. Namely, let us consider a function $w:Pr\to \mathcal{P}(\Sigma^\ast)$ called \emph{a valuation}; this function assigns a formal language over some fixed alphabet $\Sigma$ to each primitive type. It can be extended to types and sequents according to principles stated earlier; namely, $\overline{w}$ is defined as follows:
		\begin{multicols}{2}
			\begin{enumerate}
				\item $\overline{w}(B\backslash A)=\{u\in\Sigma^\ast\mid \forall v\in \overline{w}(B) \; vu\in \overline{w}(A)\}$;
				\item $\overline{w}(A/B)=\{u\in\Sigma^\ast\mid \forall v\in \overline{w}(B) \; uv\in \overline{w}(A)\}$;
				\item $\overline{w}(A\cdot B)=\{uv \mid u \in \overline{w}(A), v \in \overline{w}(B)\}$;
				\item $\overline{w}(A_1,\dots,A_n)=\overline{w}(A_1\cdot\dotsc\cdot A_n)$;
				\item $\overline{w}(\Pi\to A)$ is true if and only if $\overline{w}(\Pi)\subseteq \overline{w}(A)$.
			\end{enumerate}
		\end{multicols}
	Pentus \cite{Pentus95} proved that $\mathrm{L}\vdash \Pi\to A$ if and only if $\overline{w}(\Pi\to A)$ is true for all valuations $w$. This allows us to call the Lambek calculus the logic of formal languages in the sense that it describes all the true facts about formal languages in the signature $\backslash, \cdot, /, \subseteq$ and only them.
	
	It is often important to work with more complex structures than strings. This is the reason why a wide variety of extensions of generative grammars to terms, graphs etc. has been introduced. In this paper, we focus on a particular approach called \emph{hyperedge replacement grammars (HRGs)}. A survey on HRGs can be found in \cite{Drewes97, Habel92}; in this paper, we mainly follow the definitions and notation from \cite{Drewes97}. We chose HRGs as the basis for our studies since they are very close to context-free grammars in the sense of definitions, underlying mechanisms and properties. Our main goal is to extend the Lambek calculus and Lambek grammars to hypergraphs in a natural way and to study the resulting notions.
	
	In Section \ref{sec_prelim}, we define fundamental notions regarding hypergraphs and hyperedge replacement. In Section \ref{sec_HL}, we introduce the hypergraph Lambek calculus extending notions of types, sequents, axioms and rules. The formal definition of hypergraph Lambek grammars will be given in Section \ref{sec_HLG}. There we also study the power of these grammars; it turns out that they can generate more languages than HRGs, e.g. the language of all graphs, the language of all bipartite graphs, the language of all regular graphs. Since the membership problem for $\mathrm{HL}$-grammars is NP-complete and since they generate all isolated-node bounded languages generated by HRGs, they can be considered as an attractive alternative to HRGs. In Section \ref{sec_conclusion}, we conclude and outline further research directions regarding $\mathrm{HL}$ and $\mathrm{HL}$-grammars.

	\section{Hyperedge Replacement}\label{sec_prelim}
	Formal definitions of hypergraphs and of hyperedge replacement are given in this section according to \cite[Chapter 3]{Drewes97}.
	
	$\mathbb{N}$ includes $0$. $\Sigma^\ast$ is the set of all strings over the alphabet $\Sigma$ including the empty string $\Lambda$. $\Sigma^\circledast$ is the set of all strings consisting of distinct symbols (they may be considered as ordered sets). The length $|w|$ of the word $w$ is the number of symbols in $w$. The set of all symbols contained in a word $w$ is denoted by $[w]$. If $f:\Sigma\to\Delta$ is a function from one set to another, then it is naturally extended to a function $f:\Sigma^*\to\Delta^*$ ($f(\sigma_1\dots\sigma_k)=f(\sigma_1)\dots f(\sigma_k)$).
	
	\begin{definition}
		A ranked set is a set $C$ along with the ranking function $rank: C\to \mathbb{N}$.
	\end{definition}
	\begin{definition}\label{def_hypergraph}
		\emph{A hypergraph $G$ over $C$} is a tuple $G=\langle V, E, att, lab, ext \rangle$ where $V$ is a set of \emph{nodes}, $E$ is a set of \emph{hyperedges}, $att: E\to V^\circledast$ assigns an ordered set of \emph{attachment} nodes to each edge, $lab: E \to C$ labels each edge by some element of $C$ in such a way that $rank(lab(e))=|att(e)|$ whenever $e\in E$, and $ext\in V^\circledast$ is an ordered set of \emph{external} nodes. 
		
		Components of a hypergraph $G$ are denoted by $V_G, E_G, att_G, lab_G, ext_G$ resp.
	\end{definition}
	In the remainder of the paper, hypergraphs are usually called just graphs, and hyperedges are called edges. The set of all graphs with labels from $C$ is denoted by $\mathcal{H}(C)$. Graphs are usually named by letters $G$ and $H$.
	
	In drawings of graphs, black dots correspond to nodes, labeled squares correspond to edges, $att$ is represented by numbered lines, and external nodes are depicted by numbers in brackets. If an edge has exactly two attachment nodes, it can be depicted by an arrow (which goes from the first attachment node to the second one).
	
	Note that Definition \ref{def_hypergraph} implies that attachment nodes of each hyperedge are distinct, and so are external nodes. This restriction can be removed (i.e. we can consider graphs with loops), and all definitions will be preserved; however, in this paper, we stick to the above definition.
	\begin{definition}\label{def_rank}
		The function $rank_G$ (or $rank$, if $G$ is clear) returns the number of nodes attached to an edge in a graph $G$: $rank_G(e)\eqdef|att_G(e)|$.
		If $G$ is a graph, then $rank(G)\eqdef |ext_G|$.
	\end{definition}
	A sub-hypergraph (or just subgraph) $H$ of a graph $G$ is a hypergraph such that $V_H\subseteq V_G$, $E_H\subseteq E_G$, and for all $e\in E_H$ $att_H(e)=att_G(e)$, $lab_H(e)=lab_G(e)$.
	
	If $H=\langle \{v_i\}_{i=1}^n,\{e_0\},att,lab,v_1\dots v_n\rangle$, $att(e_0)=v_1\dots v_n$ and $lab(e_0)=a$, then $H$ is called \emph{a handle}. In this work, we denote it by $a^\bullet$.
	
	\emph{An isomorphism} between graphs $G$ and $H$ is a pair of bijective functions $\mathcal{E}: E_G\to E_H$, $\mathcal{V}: V_G\to V_H$ such that $att_H\circ\mathcal{E}=\mathcal{V}\circ att_G$, $lab_G=lab_H\circ\mathcal{E}$, $\mathcal{V}(ext_G)=ext_H$. In this work, we do not distinguish between isomorphic graphs.
	
	Strings can be considered as graphs with the string structure. This is formalized in
	\begin{definition}\label{def_str_gr}
		\emph{A string graph} induced by a string $w=a_1\dots a_n$ is a graph of the form $\langle \{v_i\}_{i=0}^n,\{e_i\}_{i=1}^n,att,\\lab,v_0v_n \rangle$ where $att(e_i)=v_{i-1}v_i$, $lab(e_i)=a_i$. It is denoted by $\SG(w)$.
	\end{definition}
	We additionally introduce the following definition (not from \cite{Drewes97}):
	\begin{definition}
		Let $H\in \mathcal{H}(C)$ be a graph. Let $f:E_H\to C$ be a function, where $rank(lab_H(e))=rank(f(e))$ for $e\in E_H$. Then $f(H)=\langle V_H, E_H, att_H,\\ f, ext_H\rangle$.
	\end{definition}
	\noindent
	If one wants to relabel only one edge $e_0$ within $H$ with a label $a$, then the result is denoted by $H[e_0\eqdef a]$.
	
	\begin{definition}\label{ssec_repl}
		\emph{Hyperedge replacement} is defined in \cite{Drewes97}, and it plays a fundamental role in hyperedge replacement grammars. The replacement of an edge $e_0$ in $G$ with a graph $H$ can be done if $rank(e_0)=rank(H)$ as follows:
		\begin{enumerate}
			\item Remove $e_0$;
			\item Insert an isomorphic copy of $H$ ($H$ and $G$ have to consist of disjoint sets of nodes and edges);
			\item For each $i$ in $1..rank_G(e_0)$, fuse the $i$-th external node of $H$ with the $i$-th attachment node of $e_0$.
		\end{enumerate}
		The result is denoted by $G[e_0/H]$. 
	\end{definition}
	It is known that if several edges of a graph are replaced by other graphs, then the result does not depend on the order of replacements; moreover the result is not changed if replacements are done simultaneously (see \cite{Drewes97}). The following notation is in use: if $e_1,\dots,e_k$ are distinct edges of a graph $H$ and they are simultaneously replaced by graphs $H_1,\dots,H_k$ resp. (this requires $rank(H_i)=rank(e_i)$), then the result is denoted $H[e_1/H_1,\dots,e_k/H_k]$.
	
	\section{Hypergraph Lambek Calculus}\label{sec_HL}
	HRGs can be used to describe linguistic structures as well as context-free grammars since linguistic objects often have an underlying structure, which is more complex than a string. One of the recently studied applications is using HRGs for abstract meaning representation (see e.g. \cite{Drewes17,Gilroy17,Jones12}): the meaning of a sentence is represented by a graph.  In \cite{Bauer16}, HRGs are used for modelling nonprojective dependencies in Dutch. Another example where graph structures occur in linguistics is syntactic trees. Given a context-free grammar, it is natural to represent an internal hierarchical structure of constituents of a generated sentence by a tree.
	\begin{example}\label{ex_hrg_tree}
		The HRG $HGr = \langle\{S,\mathit{NP},N\},\{\mbox{the},\mbox{cat},\mbox{sleeps},l,r\},P,S\rangle$ with the list of productions $P$ defined below generates the graph $\mathit{Syntree}$:
		\begin{center}
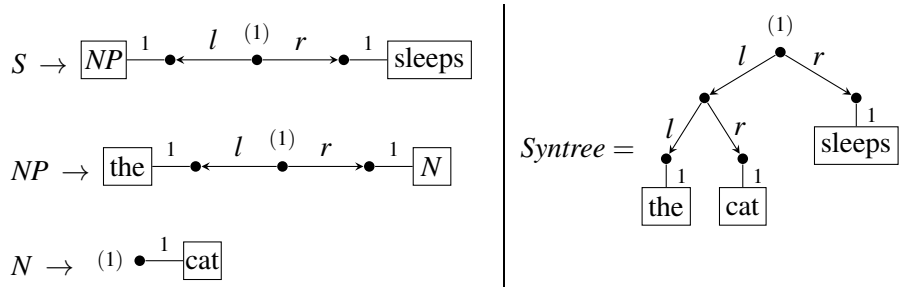
		\begin{tabular}{c|c}
			\begin{tabular}{l}
				$
				S\;\to\; {\tikz[baseline=.1ex]{
						\node (V) {};
						\node[hyperedge,above=-2mm of V] (E1) {$\:\mathit{NP}\:$};
						\node[node,right=5mm of E1] (N1) {};
						\node[node,right=10mm of N1, label=above:{\scriptsize $(1)$}] (N2) {};
						\node[node,right=10mm of N2] (N3) {};
						\node[hyperedge,right=5mm of N3] (E2) {\:sleeps\:};
						\draw[-,black] (E1) -- node[above] {\scriptsize 1} (N1);
						\draw[-,black] (N3) -- node[above] {\scriptsize 1} (E2);
						\draw[>=stealth,->,black] (N2) -- node[above] {$l$} (N1);
						\draw[>=stealth,->,black] (N2) -- node[above] {$r$} (N3);
				}}
				$
				\\
				\\
				$
				\mathit{NP}\;\to\; {\tikz[baseline=.1ex]{
						\node (V) {};
						\node[hyperedge,above=-2mm of V] (E1) {\:the\:};
						\node[node,right=5mm of E1] (N1) {};
						\node[node,right=10mm of N1, label=above:{\scriptsize $(1)$}] (N2) {};
						\node[node,right=10mm of N2] (N3) {};
						\node[hyperedge,right=5mm of N3] (E2) {$\:N\:$};
						\draw[-,black] (E1) -- node[above] {\scriptsize 1} (N1);
						\draw[-,black] (N3) -- node[above] {\scriptsize 1} (E2);
						\draw[>=stealth,->,black] (N2) -- node[above] {$l$} (N1);
						\draw[>=stealth,->,black] (N2) -- node[above] {$r$} (N3);
				}}
				$
				\\
				\\
				$N\;\to\;{\tikz[baseline=.1ex]{
						\node (V) {};
						\node[node,above=0mm of V, label=left:{\scriptsize $(1)$}] (N1) {};
						\node[hyperedge,right=5mm of N1] (E1) { cat };
						\draw[-,black] (E1) -- node[above] {\scriptsize 1} (N1);
				}}$
			\end{tabular}
			&
			$\mathit{Syntree}={\tikz[baseline=.1ex]{
					\node (V) {};
					\node[node, above=12mm of V, label=above:{\scriptsize $(1)$}] (N1) {};
					\node[node, below left=5mm and 9mm of N1] (N2) {};
					\node[node, below right=5mm and 9mm of N1] (N3) {};
					\node[node, below left=7mm and 4mm of N2] (N4) {};
					\node[node, below right=7mm and 4mm of N2] (N5) {};
					\node[hyperedge,below=3mm of N3] (E1) {\:sleeps\:};
					\node[hyperedge,below=3mm of N4] (E2) {\:the\:};
					\node[hyperedge,below=3mm of N5] (E3) {\:cat\:};
					\draw[-,black] (N3) -- node[right] {\scriptsize 1} (E1);
					\draw[-,black] (N4) -- node[right] {\scriptsize 1} (E2);
					\draw[-,black] (N5) -- node[right] {\scriptsize 1} (E3);
					\draw[->,black] (N1) -- node[above] {$l$} (N2);
					\draw[->,black] (N1) -- node[above] {$r$} (N3);
					\draw[->,black] (N2) -- node[left] {$l$} (N4);
					\draw[->,black] (N2) -- node[right] {$r$} (N5);
			}}$
			\\
		\end{tabular}
		\end{center}
		Here $rank(S)=rank(NP)=rank(N)=rank(\mbox{sleeps})=rank(\mbox{the})=rank(\mbox{cat})=1$, and $rank(l)=\\=rank(r)=2$.
		
		$\mathit{Syntree}$ is a simplified syntactic tree for the sentence \textit{the cat sleeps}; $l$ and $r$ distinguish left and right children in the tree.
	\end{example}
	Since there are such cases in linguistics where we need to work with graphs rather than with strings, we would like to generalize the categorial point of view discussed in Section \ref{sec_intr_cfg_lg} to hypergraphs. Our first attempt was a generalization of basic categorial grammars to hypergraphs; the resulting formalism called \textit{hypergraph basic categorial grammar} was introduced at ICGT 2020 \cite{Pshenitsyn20_2}. However, this formalism did not significantly improve our insight since most results were proved in the same way as for strings; in particular, such grammars generate the same set of languages as HRGs (with some nonsubstantial exceptions related to the number of isolated nodes). In this paper, we aim to go further and to discuss how the Lambek calculus along with its grammars can be generalized to hypergraphs. Note that, in the string case, Lambek grammars (i.e. grammars based on the Lambek calculus) are equivalent to context-free grammars and to basic categorial grammars. (This nontrivial result was proved in \cite{Pentus93}.)
	
	It is known that strings can be represented as string graphs, e.g. 	{\tikz[baseline=.1ex]{
				\node[node,label=left:{\scriptsize $(1)$}] (N1) {};
				\node[node,right=12mm of N1] (N2) {};
				\node[node,right=12mm of N2] (N3) {};
				\node[node,right=12mm of N3,label=right:{\scriptsize $(2)$}] (N4) {};
				\draw[->,black] (N1) -- node[above] {the} (N2);
				\draw[->,black] (N2) -- node[above] {cat} (N3);
				\draw[->,black] (N3) -- node[above] {sleeps} (N4);
	}} represents the string \textit{the cat sleeps} (this graph is denoted as $\SG(\mbox{the cat sleeps})$). The production $S\to \mathit{NP} \mbox{ sleeps}$ then is transformed into the graph production $S\to \SG(\mathit{NP} \mbox{ sleeps})$. We want to transform this production into a correspondence as it was done in Section \ref{sec_intr_cfg_lg}; there we ``took out'' the terminal unit \textit{sleeps} from the right-hand side, and assigned the type $\mathit{NP}\backslash S$ to it. Now, we are going to do the same operation but we shall mark an edge, from which we took out the label \textit{sleeps}, by a special \$ symbol:
	$$
	S\to \SG(\mathit{NP} \mbox{ sleeps}) \quad\rightsquigarrow\quad \mbox{sleeps} \corr S/\SG(\mathit{NP} \;\$)
	$$
	In general, \$ denotes the hyperedge, from which we took out its label. Similarly, the production $[\mathit{NP}\to\mbox{the } N]$ is transformed as follows: $\mbox{the} \corr \mathit{NP}/\SG(\$\;N)$. Note that now we do not need two divisions $\backslash$ and $/$ anymore since the difference between them is now expressed by the position of the \$-labeled edge. In order to distinguish the string divisions and the new graph division, and also to stress that the latter is undirected, we write $A\div D$ instead of $A/D$ for the latter. 
	
	The conversion of a production into a correspondence between a terminal unit (a symbol, a label, a word) and a type requires that there is exactly one terminal unit in the right-hand side of the production. This property is called weak Greibach normal form. Note that, if we have e.g. a production $S\to  \SG(ABaCD)$ where $S,A,B,C,D$ are nonterminal, and $a$ is terminal, then we can also use the above \$-notation and write $a\corr S\div \SG(AB\$CD)$. As in Section \ref{sec_intr_cfg_lg}, this means that $a$ is such an object that whenever objects of types $A$, $B$, $C$ and $D$ are placed instead of corresponding edges in the graph $\SG(AB\$CD)$, and $a$ is placed on the \$-labeled edge, the resulting structure forms an object of the type $S$. We can proceed similarly with an arbitrary hypergraph production, if the grammar is in the weak Greibach normal form. This is the main idea of hypergraph basic categorial grammars. However, as in the string case, we would like to go further and to consider more complex types; besides, we would also like to generalize the operation $A \cdot B$. This results in the following
	\begin{definition}\label{def_type}
		The set $\mathit{Tp}(\mathrm{HL})$ of types of the hypergraph Lambek calculus $\mathrm{HL}$ is defined inductively as the least set satisfying the following conditions:
		\begin{enumerate}
			\item Primitive types $Pr$ are types. $Pr$ is a countable ranked set such that for each $n$ there are infinitely many types $p\in Pr$ such that $rank(p)=n$.
			\item\label{def_type_div} Let $N$ (``numerator'') be a type. Let $D$ (``denominator'') be a hypergraph such that exactly one of its hyperedges (call it $e_0$) is labeled by $\$$, and the other hyperedges (possibly, there are none of them) are labeled by elements of $\mathit{Tp}(\mathrm{HL})$; let also $rank(N)=rank(D)$. Then $T=(N\div D)$ is a type, and $rank(T)\eqdef rank_D(e_0)$.
			\item\label{def_type_times} Let $M$ be a hypergraph such that all its hyperedges are labeled by types from $\mathit{Tp}(\mathrm{HL})$ (possibly, there are no hyperedges at all). Then $T=\times(M)$ is also a type, and $rank(T)\eqdef rank(M)$.
		\end{enumerate}
	\end{definition}
	Below we often write $N\div (D)$ or $N\div D$ instead of $(N\div D)$. Item \ref{def_type_div} generalizes the concept of $\div$ explained earlier: $N\div D$ is the type of such hypergraphs $H$ that, if we replace the \$-labeled edge in $D$ by $H$ and for all the remaining edges $e_i,i>0$, which are labeled by some types $T_i$, we replace them by hypergraphs $H_i$, which are of types $T_i$, then we obtain a hypergraph of the type $N$. In particular, this explains why we require $rank(N)=rank(D)$ and $rank(N\div D)=rank_D(e_0)$.
	\begin{example}\label{ex_prodtorel}
		The first production from Example \ref{ex_hrg_tree} can be transformed into the following correspondence:
		\begin{equation}\label{eq_sleeps_tree}
			\mbox{sleeps} \; \corr \; s\div \left({\tikz[baseline=.1ex]{
					\node (V) {};
					\node[hyperedge,above=-3mm of V] (E1) {$\:\mathit{np}\:$};
					\node[node,right=5mm of E1] (N1) {};
					\node[node,right=10mm of N1, label=above:{\scriptsize $(1)$}] (N2) {};
					\node[node,right=10mm of N2] (N3) {};
					\node[hyperedge,right=5mm of N3] (E2) {\$};
					\draw[-,black] (E1) -- node[above] {\scriptsize 1} (N1);
					\draw[-,black] (N3) -- node[above] {\scriptsize 1} (E2);
					\draw[>=stealth,->,black] (N2) -- node[above] {$p_l$} (N1);
					\draw[>=stealth,->,black] (N2) -- node[above] {$p_r$} (N3);
			}}\right)
		\end{equation}
		Here $p_l$, $p_r$ are primitive types introduced to deal with special ``technical'' labels $l$ ($r$ resp.). According to our general understanding of types, one may say that $l$ ($r$) is of the type $p_l$ ($p_r$ resp.). Now, (\ref{eq_sleeps_tree}) means that \textit{sleeps} is such an object that, if we place it instead of \$ within the graph {\tikz[baseline=.1ex]{
				\node (V) {};
				\node[hyperedge,above=-2mm of V] (E1) {$\:\mathit{np}\:$};
				\node[node,right=5mm of E1] (N1) {};
				\node[node,right=10mm of N1, label=above:{\scriptsize $(1)$}] (N2) {};
				\node[node,right=10mm of N2] (N3) {};
				\node[hyperedge,right=5mm of N3] (E2) {\$};
				\draw[-,black] (E1) -- node[above] {\scriptsize 1} (N1);
				\draw[-,black] (N3) -- node[above] {\scriptsize 1} (E2);
				\draw[>=stealth,->,black] (N2) -- node[above] {$p_l$} (N1);
				\draw[>=stealth,->,black] (N2) -- node[above] {$p_r$} (N3);
		}}, and we place any object (syntactic tree) of the type $\mathit{np}$ instead of the $\mathit{np}$-labeled edge, then we obtain an object of the type $s$ (we also need to replace $p_l$ by $l$ and $p_r$ by $r$).
	\end{example}
	The operation $\times(M)$ can be called a hypergraph product, or a hypergraph concatenation. Its general semantics it the following: if $E_M=\{m_1,\dots, m_l\}$, $lab_M(m_i)=T_i$, and $H_i$ is a hypergraph of the type $T_i$ ($i=1,\dots, l$), then the hypergraph $M[m_1/H_1,\dots,m_l/H_l]$ is a hypergraph of the type $\times(M)$. Thus, $\times(M)$ is the type of all substitution instances of $M$.
	\begin{example}\label{ex_strstr}
		If $str$ is the primitive type of all string graphs labeled by the blank symbol $\ast$, then the type $\times\left({\tikz[baseline=.1ex]{
				\node (V) {};
				\node[node, above=-1mm of V, label=left:{\scriptsize $(1)$}] (N1) {};
				\node[node,right=12mm of N1, label=right:{\scriptsize $(2)$}] (N2) {};
				\draw[>=stealth,->,black] (N1) to[bend left=30] node[above] {$str$} (N2);
				\draw[>=stealth,->,black] (N1) to[bend right=30] node[below] {$str$} (N2);
		}}\right)$ is the type of all graphs consisting of two parallel strings with the common start and finish nodes, e.g. {\tikz[baseline=.1ex]{
		\node (V) {};
		\node[node, above=-1mm of V, label=left:{\scriptsize $(1)$}] (N1) {};
		\node[node,right=24mm of N1, label=right:{\scriptsize $(2)$}] (N2) {};
		\node[node,above right= 1.5mm and 8mm of N1] (U1) {};
		\node[node,right= 6.6mm of U1] (U2) {};
		\node[node,below right= 1.5mm and 12mm of N1] (B1) {};

		\draw[>=stealth,->,black] (N1) to[bend left=0] (U1);
		\draw[>=stealth,->,black] (U1) to[bend right=0] (U2);
		\draw[>=stealth,->,black] (U2) to[bend right=0] (N2);
		\draw[>=stealth,->,black] (N1) to[bend right=0] (B1);
		\draw[>=stealth,->,black] (B1) to[bend right=0] (N2);
	}}.
	\end{example}
	Moving away from the intuition of $\div$ and $\times$, we would like to introduce a syntactic \emph{calculus}, which would work with types introduced in Definition \ref{def_type} by means of axioms and rules. We expect that this calculus should generalize the Lambek calculus in the Gentzen style introduced in Section \ref{sec_intr_cfg_lg}. This is done in our preprint \cite{Pshenitsyn21}; in this paper, we only introduce the axiom and rules of the hypergraph Lambek calculus ($\mathrm{HL}$) without a detailed discussion of why they actually generalize those of $\mathrm{L}$.
	\begin{definition}
		\emph{A hypergraph sequent} is a structure of the form $H\to A$, where $A\in \mathit{Tp}(\mathrm{HL})$ is a type, $H\in\mathcal{H}(\mathit{Tp}(\mathrm{HL}))$ is a hypergraph labeled by types and $rank(H)=rank(A)$. $H$ is called the \emph{antecedent} of the sequent, and $A$ is called the \emph{succedent} of the sequent.
	\end{definition}
	\begin{remark}
		Returning to our intuition, $H\to A$ could be understood as the statement ``each hypergraph of type $H$ is also of type $A$'' (in a similar way as for $\mathrm{L}$ that we described in Section \ref{sec_intr_cfg_lg}). However, we do not have types of the form $H$ for $H$ being a hypergraphs. Note that there we defined $\overline{w}(A_1,\dotsc,A_n)$ as $\overline{w}(A_1\cdot\dotsc\cdot A_n)$ in the string case. This allows us to conclude that here we should also understand $H\to A$ as the statement ``each hypergraph of type $\times(H)$ is also of type $A$''
	\end{remark}
	The hypergraph Lambek calculus $\mathrm{HL}$ deals with hypergraph sequents and explains, which of them are \textit{derivable} using axioms and rules. The only \textbf{axiom} of $\mathrm{HL}$ is the following: $p^\bullet\to p,\quad p\in Pr$ ($p^\bullet$ here is the $p$-handle). Rules are presented below along with some simple examples.
	\begin{enumerate}
		\item\label{subsec_div_to} \textbf{Rule $(\div\to)$.} Let $N\div D$ be a type and let $E_D=\{d_0,d_1,\dots,d_k\}$ where $lab(d_0)=\$$, $lab(d_i)=T_i$ for $i\ge 1$. Let $H\to A$ be a hypergraph sequent and let $e\in E_H$ be labeled by $N$. Let finally $H_1,\dots,H_k$ be hypergraphs labeled by types. Then the rule $(\div\to)$ is the following:
		$$
		\infer[(\div\to)]{H[e/D][d_0\eqdef N\div D][d_1/H_1,\dots,d_k/H_k]\to A}{H\to A & H_1\to T_1 &\dots & H_k\to T_k}
		$$
		This rule explains how a type with division may appear in the antecedent of a sequent: we replace a hyperedge $e$ by $D$, put a label $N\div D$ instead of \$ and replace the remaining labels of $D$ by corresponding antecedents. 
		\begin{example}\label{ex_div_to}
			Consider the following rule application with $T_i$ being some types and with $T$ being equal to $q\div\SG(T_2\$T_3)$:
			$$
			\infer[(\div\to)]{
				\SG(prs\,T\,tu)\to T_1
			}{
				\SG(pq)\to T_1 & \SG(rs)\to T_2 & \SG(tu) \to T_3 
			}
			$$
		\end{example}
		\item \textbf{Rule $(\to\div)$.} Let $F\to N\div D$ be a hypergraph sequent; let $e_0\in E_D$ be labeled by \$. Then
		$$
		\infer[(\to\div)]{F\to N\div D}{D[e_0/F]\to N}
		$$
		This rule is understood as follows: if there are such hypergraphs $D,F$ and such a type $N$ that in a sequent $H\to N$ the hypergraph $H$ equals $D[e_0/F]$ and $H\to N$ is derivable, then $F\to N\div D$ is also derivable.
		\begin{example}\label{ex_to_div}
			Consider the following rule application where $N$ equals $\times (\SG(pqr))$ (here we draw string graphs instead of writing $\SG(w)$ to visualize the rule application):
			$$
			\infer[(\to\div)]{
				\mbox{	
					{\tikz[baseline=.1ex]{
							\node[node,label=left:{\scriptsize $(1)$}] (N1) {};
							\node[node,right=8mm of N1] (N2) {};
							\node[node,right=8mm of N2,label=right:{\scriptsize $(2)$}] (N3) {};
							\draw[->,black] (N1) -- node[above] {$p$} (N2);
							\draw[->,black] (N2) -- node[above] {$q$} (N3);
				}}}\to N\div\left(\mbox{	
					{\tikz[baseline=.1ex]{
							\node[node,label=left:{\scriptsize $(1)$}] (N1) {};
							\node[node,right=8mm of N1] (N2) {};
							\node[node,right=8mm of N2,label=right:{\scriptsize $(2)$}] (N3) {};
							\draw[->,black] (N1) -- node[above] {$\$$} (N2);
							\draw[->,black] (N2) -- node[above] {$r$} (N3);
				}}}\right)
			}{
				\mbox{	
					{\tikz[baseline=.1ex]{
							\node[node,label=left:{\scriptsize $(1)$}] (N1) {};
							\node[node,right=8mm of N1] (N2) {};
							\node[node,right=8mm of N2] (N3) {};
							\node[node,right=8mm of N3,label=right:{\scriptsize $(2)$}] (N4) {};
							\draw[->,black] (N1) -- node[above] {$p$} (N2);
							\draw[->,black] (N2) -- node[above] {$q$} (N3);
							\draw[->,black] (N3) -- node[above] {$r$} (N4);
				}}}\to N}
			$$
		\end{example}
		\item \textbf{Rule $(\times\to)$.} Let $G\to A$ be a hypergraph sequent and let $e\in E_G$ be labeled by $\times(F)$. Then
		$$
		\infer[(\times\to)]{G\to A}{G[e/F]\to A}
		$$
		This rule is formulated from bottom to top as the previous one. Intuitively speaking, there is a subgraph of an antecedent in a premise, and it is ``compressed'' into a single $\times(F)$-labeled hyperedge.
		\begin{example}\label{ex_times_to}
			Consider the following rule application where $U$ equals $\times(\SG(pqrs))$:
			$$
			\infer[(\times\to)]{
				\mbox{	
					{\tikz[baseline=.1ex]{
							\node[node,label=left:{\scriptsize $(1)$}] (N1) {};
							\node[node,right=8mm of N1] (N2) {};
							\node[node,right=17.3mm of N2] (N4) {};
							\node[node,right=8mm of N4,label=right:{\scriptsize $(2)$}] (N5) {};
							\draw[->,black] (N1) -- node[above] {$p$} (N2);
							\draw[->,black] (N2) -- node[above] {$\times(\SG(qr))$} (N4);
							\draw[->,black] (N4) -- node[above] {$s$} (N5);	
				}}}\to U
			}{
				\mbox{	
					{\tikz[baseline=.1ex]{
							\node[node,label=left:{\scriptsize $(1)$}] (N1) {};
							\node[node,right=8mm of N1] (N2) {};
							\node[node,right=8mm of N2] (N3) {};
							\node[node,right=8mm of N3] (N4) {};
							\node[node,right=8mm of N4,label=right:{\scriptsize $(2)$}] (N5) {};
							\draw[->,black] (N1) -- node[above] {$p$} (N2);
							\draw[->,black] (N2) -- node[above] {$q$} (N3);
							\draw[->,black] (N3) -- node[above] {$r$} (N4);
							\draw[->,black] (N4) -- node[above] {$s$} (N5);	
				}}}\to U}
			$$
		\end{example}
		\item \textbf{Rule $(\to\times)$.} Let $\times(M)$ be a type and let $E_M=\{m_1,\dots,m_l\}$. Let $H_1,\dots,H_l$ be graphs. Then
		$$
		\infer[(\to\times)]{M[m_1/H_1,\dots,m_l/H_l]\to\times(M)}{H_1\to lab(m_1) & \dots & H_l\to lab(m_l)}
		$$
		This rule is quite intuitive: several sequents can be combined into a single one via some hypergraph structure $M$.
		\begin{example}\label{ex_to_times}
			Consider the following rule application with $T_i$ being some types:
			$$
			\infer[(\to\times)]{
				\SG(pqrstu)\to \times(\SG(T_1T_2T_3))
			}{
				\SG(pq) \to T_1 & \SG(rs) \to T_2 & \SG(tu) \to T_3 
			}
			$$
		\end{example}
	\end{enumerate}
	\begin{definition}
		A hypergraph sequent $H\to A$ is \emph{derivable in $\mathrm{HL}$}, written as $\mathrm{HL}\vdash H\to A$, if it can be obtained from axioms using rules of $\mathrm{HL}$. A corresponding sequence of rule applications is called \emph{a derivation} and its representation as a tree is called \emph{a derivation tree}.
	\end{definition}
	\begin{example}
		Let $rank(s)=rank(\mathit{np})=rank(n)=1$, $rank(p_l)=rank(p_r)=2$, and let
		\begin{multicols}{2}
		$V=s\div \left({\tikz[baseline=.1ex]{
					\node (V) {};
					\node[hyperedge,above=-3mm of V] (E1) {$\:\mathit{np}\:$};
					\node[node,right=5mm of E1] (N1) {};
					\node[node,right=10mm of N1, label=above:{\scriptsize $(1)$}] (N2) {};
					\node[node,right=10mm of N2] (N3) {};
					\node[hyperedge,right=5mm of N3] (E2) {\$};
					\draw[-,black] (E1) -- node[above] {\scriptsize 1} (N1);
					\draw[-,black] (N3) -- node[above] {\scriptsize 1} (E2);
					\draw[>=stealth,->,black] (N2) -- node[above] {$p_l$} (N1);
					\draw[>=stealth,->,black] (N2) -- node[above] {$p_r$} (N3);
			}}\right)$;
		
		$\mathit{Adj}=\mathit{np}\div \left({\tikz[baseline=.1ex]{
					\node (V) {};
					\node[hyperedge,above=-3mm of V] (E1) {\$};
					\node[node,right=5mm of E1] (N1) {};
					\node[node,right=10mm of N1, label=above:{\scriptsize $(1)$}] (N2) {};
					\node[node,right=10mm of N2] (N3) {};
					\node[hyperedge,right=5mm of N3] (E2) {$n$};
					\draw[-,black] (E1) -- node[above] {\scriptsize 1} (N1);
					\draw[-,black] (N3) -- node[above] {\scriptsize 1} (E2);
					\draw[>=stealth,->,black] (N2) -- node[above] {$p_l$} (N1);
					\draw[>=stealth,->,black] (N2) -- node[above] {$p_r$} (N3);
			}}\right)$.
		\end{multicols}
		Then the following is the derivation of the below sequent:
		$$
		\infer[(\div\to)]
		{
			{\tikz[baseline=.1ex]{
					\node (V) {};
					\node[node, above=12mm of V, label=above:{\scriptsize $(1)$}] (N1) {};
					\node[node, below left=5mm and 9mm of N1] (N2) {};
					\node[node, below right=5mm and 9mm of N1] (N3) {};
					\node[node, below left=7mm and 4mm of N2] (N4) {};
					\node[node, below right=7mm and 4mm of N2] (N5) {};
					\node[hyperedge,below=3mm of N3] (E1) {$V$};
					\node[hyperedge,below=3mm of N4] (E2) {\:$\mathit{Adj}$\:};
					\node[hyperedge,below=3mm of N5] (E3) {\:$n$\:};
					\draw[-,black] (N3) -- node[right] {\scriptsize 1} (E1);
					\draw[-,black] (N4) -- node[right] {\scriptsize 1} (E2);
					\draw[-,black] (N5) -- node[right] {\scriptsize 1} (E3);
					\draw[->,black] (N1) -- node[above] {$p_l$} (N2);
					\draw[->,black] (N1) -- node[above] {$p_r$} (N3);
					\draw[->,black] (N2) -- node[left] {$p_l$} (N4);
					\draw[->,black] (N2) -- node[right] {$p_r$} (N5);
			}} \quad \to \quad s
		}
		{
			\infer[(\div\to)]{
				{\tikz[baseline=.1ex]{
						\node (V) {};
						\node[node, above=3mm of V, label=above:{\scriptsize $(1)$}] (N1) {};
						\node[node, left=9mm of N1] (N2) {};
						\node[node, right=9mm of N1] (N3) {};
						\node[hyperedge,below=3mm of N3] (E1) {$V$};
						\node[hyperedge,below=3mm of N2] (E3) {\:$\mathit{np}$\:};
						\draw[-,black] (N3) -- node[right] {\scriptsize 1} (E1);
						\draw[-,black] (N2) -- node[right] {\scriptsize 1} (E3);
						\draw[->,black] (N1) -- node[above] {$p_l$} (N2);
						\draw[->,black] (N1) -- node[above] {$p_r$} (N3);
				}} \; \to \; s
			}
			{
				s^\bullet \to s & \mathit{np}^\bullet \to \mathit{np}
				&
				p_l^\bullet \to p_l
				&
				p_r^\bullet \to p_r
			}
			&
			\quad n^\bullet \to n
			&
			p_l^\bullet \to p_l
			&
			p_r^\bullet \to p_r
		}
		$$
	\end{example}
	In \cite{Pshenitsyn21}, we show that $\mathrm{L}$ and its different variants (with modalities, with the permutation rule etc.) can be embedded in $\mathrm{HL}$; we also show that certain structural properties of $\mathrm{L}$ can be straightforwardly lifted to $\mathrm{HL}$. Hence $\mathrm{HL}$ can be considered as an appropriate extension of the Lambek calculus to hypergraphs, as desired. Note that introduction of the division $\div$ and of the product $\times$ was motivated by the intuitive understanding of types as of families of hypergraphs (i.e. hypergraph languages). Although $\mathrm{HL}$ was defined as a purely syntactic formalism that formally explains how hypergraph sequents can be rewritten, one would expect that hypergraph languages can be considered as models of $\mathrm{HL}$. In Section \ref{sec_model}, we formally define language models (L-models in short) for $\mathrm{HL}$ in a way similar to how we introduced the notion of valuation $\overline{w}$ in Section \ref{sec_intr_cfg_lg}. We establish correctness of $\mathrm{HL}$ with respect to L-models, and completeness of its $\times$-free fragment.
	
	The following statements can be proved in a similar way as for strings (see \cite{Pshenitsyn21}):
	\begin{theorem}[cut elimination]\label{th_cut}
		If $H\to A$ and $G\to B$ are derivable in $\mathrm{HL}$, and $e_0\in E_G$ is labeled by $A$, then $G[e_0/H]\to B$ is also derivable in $\mathrm{HL}$.
	\end{theorem}
	\begin{proposition}[reversibility of $(\times\to)$ and $(\to\div)$]\label{prop_reversibility}\leavevmode
		\begin{enumerate}
			\item If $\mathrm{HL}\vdash H\to C$ and $e_0\in E_H$ is labeled by $\times(M)$, then $\mathrm{HL}\vdash H[e_0/M]\to C$.
			\item If $\mathrm{HL}\vdash H\to N\div D$ and $e_0\in E_D$ is labeled by $\$$, then $\mathrm{HL}\vdash D[e_0/H]\to N$.
		\end{enumerate}
	\end{proposition}
	These statements will be used in proofs of some results in this paper.
	
	\section{Language Models for $\mathrm{HL}$}\label{sec_model}
	In previous sections, we looked at types of the Lambek calculus (either in its string or in its hypergraph versions) as hypergraph languages; divisions and product were interpreted as operations on languages. In this section, we are going to formalize this idea for $\mathrm{HL}$ in a way similar to what we have done in Section \ref{sec_intr_cfg_lg} and in Examples \ref{ex_prodtorel} and \ref{ex_strstr}.
	\begin{definition}\label{def_val_hypergraph}
		Given a ranked alphabet $\Sigma$, we call a function $w:Pr\to \mathcal{P}(\mathcal{H}(\Sigma))$ \emph{a valuation} if for each $p\in Pr$ $rank(H)=rank(p)$ whenever $H\in w(p)$. This function assigns a hypergraph language to each primitive type. Its extension $\overline{w}$ is the following function from the set of hypergraph types to $\mathcal{P}(\mathcal{H}(\Sigma))$:
		\begin{enumerate}
			\item\label{def_val_hypergraph_div} Let $N\div D$ be a type and let $E_D=\{d_0,\dots, d_k\}$, $lab_D(d_0)=\$$, $lab_D(d_i)=T_i$. Then $\overline{w}(N\div D)$ consists of all graphs $G$ such that $D[d_0/G,d_1/H_1,\dots, d_k/H_k]$ belongs to $\overline{w}(N)$ whenever $H_1\in \overline{w}(T_1)$, ..., $H_k\in \overline{w}(T_k)$.
			\item Let $\times(M)$ be a type and let $E_M=\{m_1,\dots,m_l\}$, $lab_M(m_i)=T_i$. Then $\overline{w}(\times(M))$ consists of all graphs of the form $M[m_1/H_1,\dots, m_l/H_l]$ for all $H_i\in \overline{w}(T_i)$.
			\item We additionally define $\overline{w}(H\to A)$ as the statement $\overline{w}(\times(H))\subseteq \overline{w}(A)$.
		\end{enumerate}
	\end{definition}
	Thus we defined language models (or L-models) for the hypergraph Lambek calculus. Now we formulate some standard model-theoretic results (their proof in the hypergraph case does not differ from that in the string case).
	\begin{theorem}\label{th_correctness}
		If $\mathrm{HL}\vdash H\to A$, then $\overline{w}(H\to A)$ is true for each valuation $w$.
	\end{theorem}
	This theorem is proved by a straightforward induction on the length of a derivation. The other direction (i.e. completeness) is an open question (in the string case, this direction was a hard open problem until it was proved in \cite{Pentus95}). We expect that it holds in the hypergraph case but we have no idea how to generalize the proof from \cite{Pentus95}. However, if we consider the product-free fragment of $\mathrm{HL}$ (that is, we will consider types with $\div$ only), then the completeness theorem can be easily proved using the canonical model. In the string case, the canonical model is the following: given a type $A$, we assign the set of antecedents $\Pi$ such that $\Pi \to A$ is derivable to $A$ considering types as symbols of an (infinite) alphabet.
	\begin{theorem}\label{th_completeness}
		If $\overline{w}(H\to A)$ is true for each valuation $w$, and types in $H\to A$ do not contain $\times$, then $\mathrm{HL}\vdash H\to A$.
	\end{theorem}
	\begin{proof}
		Let us denote the fragment of $\mathrm{HL}$, in which we consider only types without $\times$, as $\mathrm{HL}(\div)$. We fix the alphabet $\Sigma=\mathit{Tp}(\mathrm{HL}(\div))$ (i.e. types without $\times$ are now symbols of the alphabet) and introduce a valuation $w_0$ for all primitive types: $w_0(p)=\{G\in\mathcal{H}(\Sigma)\mid \mathrm{HL}\vdash G\to p\}$. Note that such a definition of $w_0$ can be considered not only for primitive types $p$ but for all types $T\in \mathit{Tp}(\mathrm{HL}(\div))$. We claim that $\overline{w_0}(T)=w_0(T)$ for all such types; that is, the function $\overline{w_0}$ obtained from Definition \ref{def_val_hypergraph} coincides with $w_0$. Indeed: if $T=N\div D$, $lab_D(d_0)=\$$, and $lab_D(d_i)=T_i$ ($i=1,\dots, k$), then 
		\begin{center}
			$G\in w_0(N\div D)$ $\Leftrightarrow$ $\mathrm{HL}\vdash G\to N\div D$ $\Leftrightarrow$ $\mathrm{HL}\vdash D[d_0/G]\to N$ (see Proposition \ref{prop_reversibility}) $\Leftrightarrow$ \\$\Leftrightarrow$ $\forall H_i: (\mathrm{HL} \vdash H_i \to T_i,\; i=1,\dots, k) \quad \mathrm{HL}\vdash D[d_0/G,d_1/H_1,\dots, d_k/H_k]\to N$ $\Leftrightarrow$ $G\in \overline{w_0}(N\div D)$.
		\end{center}
		The penultimate equivalence follows from the fact that $\mathrm{HL}\vdash T_i^\bullet\to T_i$ and from the cut elimination Theorem \ref{th_cut}. Since $\overline{w_0}(H\to A)$ is true, $w_0(\times(H))\subseteq w_0(A)$; $H$ belongs to $w_0(\times(H))$ ($\mathrm{HL}\vdash H\to \times(H)$), hence $H$ belongs to $w_0(A)$. By the definition of $w_0$, this yields that $\mathrm{HL}\vdash H\to A$.
	\end{proof}
	Unfortunately, a similar proof does not work for $\mathrm{HL}$ with $\times$ (like in the string case). Thus, there is much space for further investigations. Nevertheless, Theorem \ref{th_correctness} and Theorem \ref{th_completeness} partially justify the periphrastic name of the hypergraph Lambek calculus given in the title: it is a logic of hypergraph languages.

	\section{Hypergraph Lambek Grammars and Their Power}\label{sec_HLG}
	Although in this paper we devote a great deal of attention to the hypergraph Lambek calculus itself and to its model-theoretic motivation, our main goal is to study the concept of \emph{hypergraph Lambek grammars} (HL-grammars in short). They are defined in a similar way to Lambek grammars (in the string case). A grammar is essentially a finite set of correspondences of the form $a\corr T$ where $a$ is a terminal label, and $T$ is a type; besides, in a grammar some type $S$ (not necessarily primitive) is distinguished. Then a hypergraph $G$ belongs to the language generated by the grammar if we can replace labels of its hyperedges by corresponding types (let us denote the resulting graph $G^\prime$) and derive the sequent $G^\prime \to S$ in $\mathrm{HL}$. 
	
	We consider a ranked alphabet $\Sigma$.
	\begin{definition}
		A \emph{hypergraph Lambek grammar (HL-grammar)} is a tuple $HGr=\langle \Sigma, S, \corr\rangle$ where $\Sigma$ is a finite set (alphabet), $S\in \mathit{Tp}({\mathrm{HL}})$ is a distinguished type, and $\corr\;\subseteq\Sigma\times \mathit{Tp}({\mathrm{HL}})$ is a finite binary relation such that $a\corr T$ implies $rank(a)=rank(T)$.
	\end{definition}
	\begin{definition}
		\emph{The type set} of an HL-grammar $HGr=\langle \Sigma, S, \corr\rangle$ is the set $\ts(HGr)=\{T \mid \exists a\in \Sigma : a\corr T\}$.
	\end{definition}
	\begin{definition}
		\emph{The language $L(HGr)$ generated by a hypergraph Lambek grammar} $HGr=\langle \Sigma, S, \corr\rangle$ is the set of all hypergraphs $G\in\mathcal{H}(\Sigma)$ for which a function $f_G:E_G\to \mathit{Tp}(\mathrm{HL})$ exists such that:
		\begin{enumerate}
			\item $lab_G(e)\corr f_G(e)$ whenever $e\in E_G$;
			\item $\mathrm{HL}\vdash f_G(G)\to S$.
		\end{enumerate}
	\end{definition}
	\begin{example}\label{ex_sgr}
		The HL-grammar $SGr=\langle\{a,b\},s,\corr\rangle$ where $s$ is primitive, and
		\begin{itemize}
			\item $a\corr s\div \SG(\$sp)=Q$,
		\item $b\corr p$, $b\corr s$
		\end{itemize}
		generates the language $\{\SG(a^nb^{n+1}) \mid n\ge 0\}$. For example, if one wants to check that $G=\SG(aabbb)$ belongs to $L(SGr)$, he/she follows such steps:
		\begin{enumerate}
			\item Relabel each edge in $G$ in such a way that each label is replaced by a type corresponding to it. We do this as follows: $G=\SG(aabbb) \rightsquigarrow f_G(G)=\SG(QQspp)$.
			\item Consider the sequent $f_G(G)\to s$ and derive it in $\mathrm{HL}$:
			$$
			\infer[(\div\to)]
			{
				\SG(QQspp) \to s
			}
			{
				\infer[(\div\to)]
				{
					\SG(Qsp) \to s
				}
				{
					s^\bullet \to s
					&
					s^\bullet \to s
					&
					p^\bullet \to p
				}
				&
				s^\bullet \to s
				&
				p^\bullet \to p
			}
			$$
		\end{enumerate}
		Now, notice the following: if a sequent $G^\prime\to s$ is derivable, and $G^\prime$ is labeled only by types from the type set of $SGr$, then each derivation of $G^\prime\to s$ consists only of applications of $(\div\to)$. The rule $(\div\to)$ consists of several replacements in the antecedent of a sequent, and hence the grammar $SGr$ works in a way similar to the hyperedge replacement grammar with the following set of productions:
		$$
		S\to \SG(aSP)
		\qquad
		S\to \SG(b)
		\qquad
		P\to \SG(b)
		$$
		Note that the conversion of this HRG back into $SGr$ can be made according to the principles explained in Section \ref{sec_HL}. The new grammar is actually a graph variant of a context-free grammar with the productions $S\to aSP, S\to b, P\to b$, which, clearly, generates the language $\{a^nb^{n+1} \mid n\ge 0\}$.
	\end{example}
	The transformation considered in Example \ref{ex_prodtorel} and in Example \ref{ex_sgr} is possible, if there is exactly one terminal label (say $a$) in a production; then we place \$ instead of $a$, and establish a correspondence $\corr$ between $a$ and a type made on the basis of this production.
	\begin{definition}
		An HRG $HGr$ is in \emph{the weak Greibach normal form} if there is exactly one terminal edge in the right-hand side of each production.
	\end{definition}
	Denote by $isize(H)$ the number of isolated nodes in $H$.
	\begin{definition}
		A hypergraph language $L$ is isolated-node bounded if there is a constant $M>0$ such that for each $H\in L$ $isize(H)<M\cdot |E_H|$.
	\end{definition}
	In \cite{Pshenitsyn20_1} we prove the following
	\begin{theorem}\label{WGNF}
		For each HRG generating an isolated-node bounded language there is an equivalent HRG in the weak Greibach normal form.
	\end{theorem}
	Using it, we can prove the following theorem applying standard techniques.
	\begin{theorem}\label{th_hrg_hlg}
		For each HRG generating an isolated-node bounded language there is an equivalent HL-grammar.
	\end{theorem}
	The proof of this theorem can be found in \cite{Pshenitsyn21_2}. It is not, however, of interest in this paper; we formulate this theorem here only to show the reader that HL-grammars are not weaker than HRGs (isolated-node boundedness is a nonsubstantial limitation). Our objective now is to show that HL-grammars are more powerful than HRGs; to do this, we will introduce several examples of grammars generating languages that can be generated by no HRGs.

	\subsection{All Binary Graphs}\label{sec_lan_all}
	One of restrictions known for languages generated by HRGs is that they are of bounded connectivity (see \cite{Drewes97}); this follows from the pumping lemma (see \cite{Habel92}, Chapter IV.2). Consequently, no HRG can generate the set of all binary graphs (i.e. of usual graphs with edges of rank 2). This might seem unnatural because HRGs represent a context-free formalism, and the language of all binary graphs seems to be very simple and regular. Below we show that HL-grammars are powerful enough to generate such a language.
	
	Consider the language $\mathcal{L}_1$ of all binary graphs without isolated nodes (the empty graph is not included in $\mathcal{L}_1$ as well) over the alphabet $\{\ast\}$ ($rank(\ast)=2$) that are, besides, without external nodes. Consequently, each graph in this language has at least two nodes. Let $s,p$ be primitive types ($rank(s)=0$, $rank(p)=1$). Let us define the following types:
	$$Q_1=p,\quad
		Q_2=p\div \left(\mbox{
			{\tikz[baseline=.1ex]{
					\node[] (R) {};
					\node[node,below=1mm of R,label=below:{\scriptsize $(1)$}] (N1) {};
					\node[hyperedge,above=5mm of N1] (E1) {$\$$};
					\node[node,right=7mm of N1] (N2) {};
					\node[hyperedge,above=5mm of N2] (E2) {$p$};
					\draw[-,black] (E1) -- node[left] {\scriptsize 1} (N1);
					\draw[-,black] (E2) -- node[right] {\scriptsize 1} (N2);	
		}}}\right),\quad
		Q_3=s\div\left(\mbox{
			{\tikz[baseline=.1ex]{
					\node[] (R) {};
					\node[node,below=1mm of R] (N1) {};
					\node[hyperedge,above=5mm of N1] (E1) {$\$$};
					\node[node,right=7mm of N1] (N2) {};
					\node[hyperedge,above=5mm of N2] (E2) {$p$};
					\draw[-,black] (E1) -- node[left] {\scriptsize 1} (N1);
					\draw[-,black] (E2) -- node[right] {\scriptsize 1} (N2);	
		}}}\right);$$
		$$M_{11}^{ij}= \times\left(\mbox{
			{\tikz[baseline=.1ex]{
					\node[] (R) {};
					\node[node,below=-0.5mm of R,label=below:{\scriptsize $(1)$}] (N1) {};
					\node[hyperedge,above=5mm of N1] (E1) {$Q_i$};
					\node[node,right=7mm of N1,label=below:{\scriptsize $(2)$}] (N2) {};
					\node[hyperedge,above=5mm of N2] (E2) {$Q_j$};
					\draw[-,black] (E1) -- node[left] {\scriptsize 1} (N1);
					\draw[-,black] (E2) -- node[right] {\scriptsize 1} (N2);
			}}
		}\right),\;
		M_{12}^{i}= \times\left(\mbox{
			{\tikz[baseline=.1ex]{
					\node[] (R) {};
					\node[node,below=-0.5mm of R,label=below:{\scriptsize $(1)$}] (N1) {};
					\node[hyperedge,above=5mm of N1] (E1) {$Q_i$};
					\node[node,right=7mm of N1,label=below:{\scriptsize $(2)$}] (N2) {};
					\draw[-,black] (E1) -- node[left] {\scriptsize 1} (N1);
			}}
		}\right),
		M_{21}^{j}= \times\left(\mbox{
			{\tikz[baseline=.1ex]{
					\node[] (R) {};
					\node[node,below=-0.5mm of R,label=below:{\scriptsize $(1)$}] (N1) {};
					\node[node,right=7mm of N1,label=below:{\scriptsize $(2)$}] (N2) {};
					\node[hyperedge,above=5mm of N2] (E2) {$Q_j$};
					\draw[-,black] (E2) -- node[right] {\scriptsize 1} (N2);
			}}
		}\right),\;
		M_{22}= \times\left(\!\mbox{
			{\tikz[baseline=.1ex]{
					\node[] (R) {};
					\node[node,above left=0mm and 3mm of R,label=below:{\scriptsize $(1)$}] (N1) {};
					\node[node,right=7mm of N1,label=below:{\scriptsize $(2)$}] (N2) {};
			}}
		}\!\right).$$
	Consider a grammar $HGr_1=\langle\{\ast\}, s, \corr \rangle$ where $\ast\corr N$ whenever $N\in\{M_{11}^{ij},M_{12}^{i},M_{21}^{j},M_{22}|1\le i,j\le 3\}$.
	\begin{theorem}\label{prop_lan_all}
		$L(HGr_1)=\mathcal{L}_1$.
	\end{theorem}
	\begin{proof}
		To prove that $L(HGr_1)\subseteq\mathcal{L}_1$ it suffices to note that denominators of types in $\ts(HGr_1)$ do not contain isolated nodes; since isolated nodes may appear only after applications of rules $(\div\to)$ or $(\to\times)$, all graphs in $L(HGr_1)$ do not have them.
		
		The other inclusion $L(HGr_1)\supseteq\mathcal{L}_1$ is of central interest. We start with an example of a specific derivation in this grammar. After, we provide the proof in a general case, but we suppose that this example is enough to understand the construction of $HGr_1$.
		\begin{example}
			Consider a binary graph
			$$
			H=\mbox{
				{\tikz[baseline=.1ex]{
						\node[] (R) {};
						\node[node,above=4mm of R] (N1) {};
						\node[node,below=4mm of R] (N2) {};
						\node[node,right=6.93mm of R] (N3) {};
						\node[node,right=10mm of N3] (N4) {};
						\draw[>=stealth,->,black] (N1) -- node[left] {$\ast$} (N2);
						\draw[>=stealth,->,black] (N1) -- node[above right] {$\ast$} (N3);
						\draw[>=stealth,->,black] (N2) -- node[below right] {$\ast$} (N3);
						\draw[>=stealth,->,black] (N4) -- node[above] {$\ast$} (N3);
				}}
			}
			$$
			In order to check that $H$ belongs to $L(HGr_1)$ we relabel it by corresponding types as follows:
			$$
			f_H(H)=\mbox{
				{\tikz[baseline=.1ex]{
						\node[] (R) {};
						\node[node,above=4mm of R] (N1) {};
						\node[node,below=4mm of R] (N2) {};
						\node[node,right=6.93mm of R] (N3) {};
						\node[node,right=15mm of N3] (N4) {};
						\draw[>=stealth,->,black] (N1) -- node[left] {$M_{11}^{32}$} (N2);
						\draw[>=stealth,->,black] (N1) -- node[above right] {$M_{21}^{2}$} (N3);
						\draw[>=stealth,->,black] (N2) -- node[below right] {$M_{22}$} (N3);
						\draw[>=stealth,->,black] (N4) -- node[below] {$M_{12}^{1}$} (N3);
				}}
			}
			$$
			Then we check derivability of $f_H(H)\to s$ (see Figure \ref{fig_der_all}). The idea behind types $M^{\cdot,\cdot}_{ij}$ is the following: each edge is replaced in the derivation (considered from bottom to top) by a pair of hyperedges of rank 1 attached to nodes (labeled by $Q_i$). Each node is intended to be attached to exactly one hyperedge; then we need to label exactly one hyperedge by $Q_1$ (from which rule applications of $(\div\to)$ should start), exactly one hyperedge by $Q_3$, and the remaining ones by $Q_2$.
			\begin{figure}
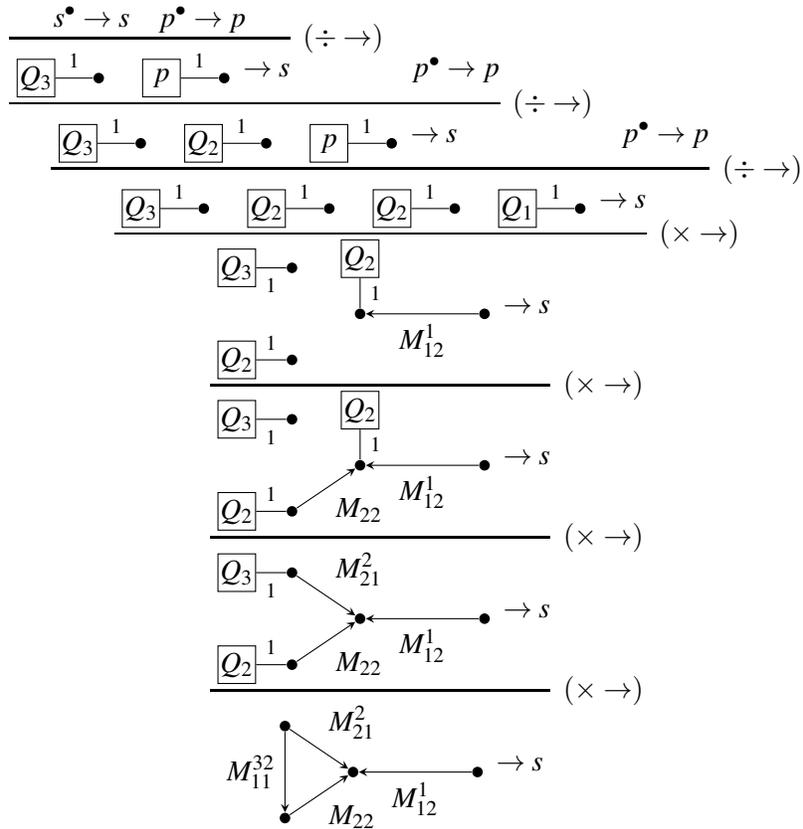

			$$
			\infer[(\times\to)]{
				\mbox{{\tikz[baseline=.1ex]{
							\node[] (R) {};
							\node[node,above=4mm of R] (N1) {};
							\node[node,below=4mm of R] (N2) {};
							\node[node,right=6.93mm of R] (N3) {};
							\node[node,right=15mm of N3] (N4) {};				
							\draw[>=stealth,->,black] (N1) -- node[left] {$M_{11}^{32}$} (N2);
							\draw[>=stealth,->,black] (N1) -- node[above right] {$M_{21}^{2}$} (N3);
							\draw[>=stealth,->,black] (N2) -- node[below right] {$M_{22}$} (N3);
							\draw[>=stealth,->,black] (N4) -- node[below] {$M_{12}^{1}$} (N3);
					}}
				}\to s}{
				\infer[(\times\to)]{
					\mbox{
						{\tikz[baseline=.1ex]{
								\node[] (R) {};
								\node[node,above=4mm of R] (N1) {};
								\node[node,below=4mm of R] (N2) {};
								\node[node,right=6.93mm of R] (N3) {};
								\node[node,right=15mm of N3] (N4) {};
								\node[hyperedge,left=4mm of N1] (E1) {$Q_3$};
								\node[hyperedge,left=4mm of N2] (E2) {$Q_2$};					
								\draw[>=stealth,->,black] (N1) -- node[above right] {$M_{21}^{2}$} (N3);
								\draw[>=stealth,->,black] (N2) -- node[below right] {$M_{22}$} (N3);
								\draw[>=stealth,->,black] (N4) -- node[below] {$M_{12}^{1}$} (N3);
								\draw[-,black] (E1) -- node[below] {\scriptsize 1} (N1);
								\draw[-,black] (E2) -- node[above] {\scriptsize 1} (N2);
						}}
					}\to s
				}{
					\infer[(\times\to)]{
						\mbox{
							{\tikz[baseline=.1ex]{
									\node[] (R) {};
									\node[node,above=4mm of R] (N1) {};
									\node[node,below=4mm of R] (N2) {};
									\node[node,right=6.93mm of R] (N3) {};
									\node[node,right=15mm of N3] (N4) {};
									\node[hyperedge,left=4mm of N1] (E1) {$Q_3$};
									\node[hyperedge,left=4mm of N2] (E2) {$Q_2$};					
									\node[hyperedge,above=4mm of N3] (E3) {$Q_2$};					
									\draw[>=stealth,->,black] (N2) -- node[below right] {$M_{22}$} (N3);
									\draw[>=stealth,->,black] (N4) -- node[below] {$M_{12}^{1}$} (N3);
									\draw[-,black] (E1) -- node[below] {\scriptsize 1} (N1);
									\draw[-,black] (E2) -- node[above] {\scriptsize 1} (N2);
									\draw[-,black] (E3) -- node[right] {\scriptsize 1} (N3);
							}}
						}\to s}{
						\infer[(\times\to)]{
							\mbox{
								{\tikz[baseline=.1ex]{
										\node[] (R) {};
										\node[node,above=4mm of R] (N1) {};
										\node[node,below=4mm of R] (N2) {};
										\node[node,right=6.93mm of R] (N3) {};
										\node[node,right=15mm of N3] (N4) {};
										\node[hyperedge,left=4mm of N1] (E1) {$Q_3$};
										\node[hyperedge,left=4mm of N2] (E2) {$Q_2$};					
										\node[hyperedge,above=4mm of N3] (E3) {$Q_2$};					
										\draw[>=stealth,->,black] (N4) -- node[below] {$M_{12}^{1}$} (N3);
										\draw[-,black] (E1) -- node[below] {\scriptsize 1} (N1);
										\draw[-,black] (E2) -- node[above] {\scriptsize 1} (N2);
										\draw[-,black] (E3) -- node[right] {\scriptsize 1} (N3);
								}}
							}\to s}{
							\infer[(\div\to)]{
								\mbox{
									{\tikz[baseline=.1ex]{
											\node[] (R) {};
											\node[hyperedge] (E1) {$Q_3$};
											\node[node,right=5mm of E1] (N1) {};
											\node[hyperedge,right=5mm of N1] (E2) {$Q_2$};					
											\node[node,right=5mm of 
											E2] (N2) {};
											\node[hyperedge,right=5mm of N2] (E3) {$Q_2$};					
											\node[node,right=5mm of E3] (N3) {};
											\node[hyperedge,right=5mm of N3] (E4) {$Q_1$};
											\node[node,right=5mm of E4] (N4) {};
											\draw[-,black] (E1) -- node[above] {\scriptsize 1} (N1);
											\draw[-,black] (E2) -- node[above] {\scriptsize 1} (N2);
											\draw[-,black] (E3) -- node[above] {\scriptsize 1} (N3);
											\draw[-,black] (E4) -- node[above] {\scriptsize 1} (N4);
									}}
								}\to s}{
								\infer[(\div\to)]{
									\mbox{
										{\tikz[baseline=.1ex]{
												\node[] (R) {};
												\node[hyperedge] (E1) {$Q_3$};
												\node[node,right=5mm of E1] (N1) {};
												\node[hyperedge,right=5mm of N1] (E2) {$Q_2$};					
												\node[node,right=5mm of 
												E2] (N2) {};
												\node[hyperedge,right=5mm of N2] (E3) {$p$};					
												\node[node,right=5mm of E3] (N3) {};
												\draw[-,black] (E1) -- node[above] {\scriptsize 1} (N1);
												\draw[-,black] (E2) -- node[above] {\scriptsize 1} (N2);
												\draw[-,black] (E3) -- node[above] {\scriptsize 1} (N3);
										}}
									}\to s}{
									\infer[(\div\to)]{
										\mbox{
											{\tikz[baseline=.1ex]{
													\node[] (R) {};
													\node[hyperedge] (E1) {$Q_3$};
													\node[node,right=5mm of E1] (N1) {};
													\node[hyperedge,right=5mm of N1] (E2) {$p$};					
													\node[node,right=5mm of 
													E2] (N2) {};
													\draw[-,black] (E1) -- node[above] {\scriptsize 1} (N1);
													\draw[-,black] (E2) -- node[above] {\scriptsize 1} (N2);
											}}
										}\to s}{s^\bullet \to s & p^\bullet \to p} & p^\bullet \to p
								} & p^\bullet \to p
							}
						}
					}
				}
			}
			$$
			\caption{Derivation of a sequent corresponding to the binary graph $H$.}\label{fig_der_all}
		\end{figure}
		\end{example}
		In general, let $H$ be in $\mathcal{L}_1$. Since there are no isolated nodes in $H$ there exists a function $h:V_H\to E_H$ such that $h(v)$ is attached to $v$ whenever $v\in V_H$. We choose two arbitrary nodes $v_b$ (begin) and $v_e$ (end) such that $v_b\ne v_e$. After that, we define a function $c:V_H\to \{1,2,3\}$ as follows: $c(v_b)=1$, $c(v_e)=3$, and $c(v)=2$ whenever $v\not\in\{v_b,v_e\}$.
		
		Now we present a relabeling $f_H:E_H\to \mathit{Tp}(\mathrm{HL})$. Let $e$ belong to $E_H$ and let $att_H(e)=v_1v_2$.
		\begin{multicols}{2}
			\begin{itemize}[leftmargin=*]
				\item If $h(v_1)=h(v_2)=e$, then $f_H(e):=M_{11}^{c(v_1)c(v_2)}$;
				\item If $h(v_1)=e,h(v_2)\ne e$, then $f_H(e):=M_{12}^{c(v_1)}$;
				\item If $h(v_1)\ne e,h(v_2)= e$, then $f_H(e):=M_{21}^{c(v_2)}$;
				\item If $h(v_1)\ne e,h(v_2)\ne e$, then $f_H(e):=M_{22}$.
			\end{itemize} 
		\end{multicols}
		We aim to check derivability of the sequent $f_H(H)\to s$. Its derivation from bottom to top starts with the rule $(\times\to)$ applied $|E_H|$ times to all types in the antecedent. It turns out that, after such applications of $(\times\to)$, the antecedent of a sequent includes one edge labeled by $Q_1$, one edge labeled by $Q_3$, and the remaining edges labeled by $Q_2$; besides, for each node there is exactly one edge attached to it (this is satisfied by the definition of the function $h$). Then we apply (again from bottom to top) the rule $(\div\to)$, and using it we ``reduce'' the only $Q_1$-labeled edge (recall that $Q_1=p$) with a $Q_2$-labeled edge; after this we obtain a new $p$-labeled edge and repeat the procedure. Thus we eliminate all nodes and edges one-by-one. Finally, we obtain a graph with two nodes, with a $Q_3$-labeled edge attached to the first one and a $p$-labeled edge attached to the second one. Applying $(\div\to)$ once more, we ``contract'' $Q_3$ with $p$ and obtain the sequent $s^\bullet \to s$, which is an axiom.
	\end{proof}
	Therefore, we have established that HL-grammars are stronger than HRGs and that they moreover disobey the pumping lemma introduced in \cite{Drewes17}.
	\subsection{Bipartite graphs}
	Another example is the language $\mathcal{L}_2\subseteq \mathcal{L}_1$ of all bipartite binary graphs without isolated nodes. 
	\begin{definition}
		A binary graph $H$ is \emph{bipartite} if its nodes can be divided into two disjoint subsets $V_1$ and $V_2$ in such a way that each edge of $H$ outgoes from a node belonging to $V_1$ to a node belonging to $V_2$.
	\end{definition}
	
	Let us define the following types (where $p,q$ are primitive, $rank(p)=rank(q)=1$):
	\begin{multicols}{2}
		\begin{itemize}
			\item $R_1(r):=r$;
			\item $R_2(r):=r\div\left(\mbox{
				{\tikz[baseline=.1ex]{
						\node[] (R) {};
						\node[node,below=-1mm of R,label=below:{\scriptsize $(1)$}] (N1) {};
						\node[hyperedge,above left = 3mm and 2mm of N1] (E1) {\$};
						\node[hyperedge,above right = 3mm and 2mm of N1] (E2) {$r$};
						\draw[-,black] (E1) -- node[below left] {\scriptsize 1} (N1);
						\draw[-,black] (E2) -- node[below right] {\scriptsize 1} (N1);
				}}
			}\right)$;
			\item $R_3(r):=r\div\left(\mbox{
				{\tikz[baseline=.1ex]{
						\node[] (R) {};
						\node[node,below=1mm of R,label=right:{\scriptsize $(1)$}] (N1) {};
						\node[hyperedge,above = 3mm of N1] (E1) {\$};
						\node[node,right=10mm of N1] (N2) {};
						\node[hyperedge,above = 3mm of N2] (E2) {$r$};
						\draw[-,black] (E1) -- node[left] {\scriptsize 1} (N1);
						\draw[-,black] (E2) -- node[left] {\scriptsize 1} (N2);
				}}
			}\right)$;
			\item $R_4(r):=r\div\left(\mbox{
				{\tikz[baseline=.1ex]{
						\node[] (R) {};
						\node[node,below=-1mm of R,label=below:{\scriptsize $(1)$}] (N1) {};
						\node[hyperedge,above left = 3mm and 2mm of N1] (E1) {\$};
						\node[hyperedge,above right = 3mm and 2mm of N1] (E2) {$r$};
						\draw[-,black] (E1) -- node[below left] {\scriptsize 1} (N1);
						\draw[-,black] (E2) -- node[below right] {\scriptsize 1} (N1);
						\node[node,right=13mm of N1] (N22) {};
						\node[hyperedge,above = 3mm of N22] (E22) {$r$};
						\draw[-,black] (E22) -- node[left] {\scriptsize 1} (N22);
				}}
			}\right)$;
			\item $M^{ij}:=\times\left(\mbox{
				{\tikz[baseline=.1ex]{
						\node[] (R) {};
						\node[node,below=-1mm of R,label=below:{\scriptsize $(1)$}] (N1) {};
						\node[hyperedge,above = 3mm of N1] (E1) {$\,R_i(p)\,$};
						\node[node,right=10mm of N1,label=below:{\scriptsize $(2)$}] (N2) {};
						\node[hyperedge,above = 3mm of N2] (E2) {$\,R_j(q)\,$};
						\draw[-,black] (E1) -- node[left] {\scriptsize 1} (N1);
						\draw[-,black] (E2) -- node[left] {\scriptsize 1} (N2);
				}}
			}\right),\quad 1\le i,j\le 4;$
			\item $S:=\times\left(\mbox{
				{\tikz[baseline=.1ex]{
						\node[] (R) {};
						\node[node,below=1mm of R] (N1) {};
						\node[hyperedge,above = 3mm of N1] (E1) {$p$};
						\node[node,right=10mm of N1] (N2) {};
						\node[hyperedge,above = 3mm of N2] (E2) {$q$};
						\draw[-,black] (E1) -- node[left] {\scriptsize 1} (N1);
						\draw[-,black] (E2) -- node[left] {\scriptsize 1} (N2);
				}}
			}\right).$
		\end{itemize}
	\end{multicols}
	We define $HGr_2:=\langle\{\ast\},S,\corr \rangle$ as follows: $\ast\corr M^{ij}$ for all $1\le i,j\le 4$.
	\begin{proposition}\label{prop_lan_bipartite}
		$\mathcal{L}_2=L(HGr_2)$.
	\end{proposition}
	The proof of this proposition is divided into two parts. It straightforward to show that $\mathcal{L}_2\subseteq L(HGr_2)$ by deriving sequents corresponding to graphs from $\mathcal{L}_2$. To prove the other inclusion, we use Proposition \ref{prop_reversibility} and then notice that there is no way for two hyperedges, one of which is labeled by $R_i(p)$ and the other one is labeled by $R_j(q)$, to be attached to the same node in the antecedent (to prove this, it suffices to analyze variants of how the rule $(\div \to)$ can be applied).
	\subsection{Regular graphs}
	A less trivial example of a hypergraph language generated by a HL-grammar and by no HRGs is the language of regular binary graphs.
	\begin{definition}
		A binary graph $H$ is \emph{regular} if there is an integer $k\ge 1$ such that the indegree and the outdegree of each node equals $k$.
	\end{definition}
	Let $\mathcal{L}_3\subseteq \mathcal{L}_1$ be the language of all regular binary graphs (without the empty graph).
	\begin{theorem}\label{th_lan_regular}
		$\mathcal{L}_3$ can be generated by some HL-grammar.
	\end{theorem}
	To prove this theorem, we need the following result proved in \cite{Pshenitsyn21_2}:
	\begin{theorem}\label{th_finite_intersections}
		If $L_1,\dots,L_k$ are languages generated by some HRGs, then $L_1\cap\dotsc\cap L_k$ can be generated by some HL-grammar.
	\end{theorem}
	Less formally, this means that HL-grammars can generate finite intersections of languages generated by HRGs.
	\begin{definition}
		Let $MS_1,\dots, MS_n$, $n\ge 1$ be some multisets with elements from $C^\prime\subseteq C$ (that is, they are multisets of labels from $C^\prime$). Let $b\in C\setminus C^\prime$ be some symbol with $rank(b)=2$. We denote $MS_i=\{a_i^1,\dots,a_i^{k_i}\}$, and $rank(a_i^j)=t_i^j>0$. \emph{A flowerbed} $\FB(MS_1,\dots,MS_n,b)$ over $C^\prime$ is the hypergraph \\ $\langle V,E,att,lab,ext\rangle$ where
		\begin{enumerate}
			\item $V=\{ v_i^{jk} \mid i=1,\dots, n,\; j=1,\dots, k_i,\; k = 1,\dots, t_i^j-1\} \cup \{u_1,\dots, u_n\}$;
			\item $E=\{e_i^j \mid i=1,\dots, n,\; j=1,\dots, k_i\}\cup \{f_1,\dots, f_{n-1}\}$;
			\item 
			\begin{enumerate}
				\item $att(e_i^j)=u_iv_i^{j1}v_i^{j2}\dots v_i^{j(t_i^j-1)}$ (if $t_i^j=1$, then $att(e_i^j)=u_i$);
				\item $att(f_k)=v_kv_{k+1}$;
			\end{enumerate}
			\item 
			\begin{enumerate}
				\item $lab(e_i^j)=a_i^j$;
				\item $lab(f_k)=b$;
			\end{enumerate}
			\item $ext=\Lambda$.
		\end{enumerate}
	\end{definition}
	Informally, a flowerbed is a string graph $\SG(b\dotsc b)$ ($b$ repeated $n$ times) without external nodes such that several hyperedges of different ranks can be additionally attached to its nodes (but only the first attachment node of a hyperedge belongs to this string graph; the remaining nodes must be attached only to this hyperedge).
	\begin{definition}
		Given a multiset $MS$, $|MS|^a$ denotes the number of occurences of $a$ in $MS$.
	\end{definition}
	\begin{proof}[Proof (of Theorem \ref{th_lan_regular})]
		Let $C^\prime=\{a,z\}$ with $rank(a)=1$, $rank(z)=3$. We set $\Sigma=\{a,z,b\}$ with $rank(b)=2$. Let us introduce the following languages:
		\begin{itemize}
			\item $L_1$ is the set of all flowerbeds over $C^\prime$ of the form $\FB(MS_1, \dots, MS_n, b)$ such that $|MS_{2k}|^a=\\=|MS_{2k+1}|^a=|MS_{2k}|^z=|MS_{2k+1}|^z$ for $k=1,\dotsc, \lfloor \frac{n-1}{2} \rfloor$
			\item $L_2$ is the set of all flowerbeds over $C^\prime$ of the form $\FB(MS_1, \dots, MS_n, b)$ such that $|MS_{2k-1}|^a=|MS_{2k}|^a=|MS_{2k-1}|^z=|MS_{2k}|^z$ for $k=1,\dotsc,\lfloor \frac n2 \rfloor$
		\end{itemize}
		It is left as an exercise to prove that $L_1$ and $L_2$ can be generated by some HRGs; thus, according to Theorem \ref{th_finite_intersections}, $L=L_1\cap L_2$ can be generated by an HL-grammar. Let us denote such a grammar $HGr=\langle \Sigma,S,\triangleright\rangle$: $L(HGr)=L$. Note that $L$ is the set of all flowerbeds over $C^\prime$ of the form $\FB(MS_1, \dots, MS_n, b)$ such that $|MS_{k}|^a=|MS_{k+1}|^a=|MS_{k}|^z=|MS_{k+1}|^z$ for $k=1,2,\dots,n-1$.
		
		Let us denote all types corresponding to $a$ via $\triangleright$ as $A_i$ (i.e. $a\corr A_i$), all types corresponding to $z$ as $Z_j$, and all types corresponding to $b$ as $B_k$. Let
		$$
		T_{ij}\eqdef \times\left(\mbox{
			{\tikz[baseline=.1ex]{
					\node[] (R) {};
					\node[node,below=-1mm of R,label=below:{\scriptsize $(1)$}] (N1) {};
					\node[hyperedge,above = 3mm of N1] (E1) {$\,A_i\,$};
					\node[node,right=15mm of N1,label=below:{\scriptsize $(2)$}] (N2) {};
					\node[hyperedge,above = 3mm of N2] (E2) {$\,Z_j\,$};
					\draw[-,black] (E1) -- node[left] {\scriptsize 1} (N1);
					\draw[-,black] (E2) -- node[left] {\scriptsize 1} (N2);
			}}
		}\right),\qquad
		T_{ijk}\eqdef \times\left(\mbox{
			{\tikz[baseline=.1ex]{
					\node[] (R) {};
					\node[node,below=-1mm of R,label=below:{\scriptsize $(1)$}] (N1) {};
					\node[hyperedge,above = 3mm of N1] (E1) {$\,A_i\,$};
					\node[node,right=15mm of N1,label=below:{\scriptsize $(2)$}] (N2) {};
					\node[hyperedge,above = 3mm of N2] (E2) {$\,Z_j\,$};
					\draw[-,black] (E1) -- node[left] {\scriptsize 1} (N1);
					\draw[-,black] (E2) -- node[left] {\scriptsize 1} (N2);
					\draw[->,black] (N1) -- node[above] {$B_k$} (N2);
			}}
		}\right).
		$$
		Using these types we define an HL-grammar $\widetilde{HGr}\eqdef \langle\{\ast\},S,\widetilde{\corr}\rangle$ as follows: $\ast \;\widetilde{\corr}\; T_{ij}, T_{ijk}$ for all possible $i$, $j$, $k$. We argue that $L(\widetilde{HGr})=L_3$. Indeed, $H\in L(\widetilde{HGr})$ if and only if there exists a relabeling $f_H$ such that $lab_H(e)\;\widetilde{\corr}\; f_H(e)$, and $\mathrm{HL}\vdash f_H(H)\to S$. Labels in $f_H(H)$ are types $T_{ij}$ and $T_{ijk}$, which are of the form $\times(M)$. Using Proposition \ref{prop_reversibility}, we draw the conclusion that $\mathrm{HL}\vdash f_H(H)\to S$ if and only if $\mathrm{HL}\vdash\hat{H}\to S$ where $\hat{H}$ is obtained from $f_H(H)$ by replacing each hyperedge labeled by a type of the form $\times(M)$ by $M$. $\hat{H}$ is labeled by types $A_i$, $Z_j$, and $B_k$. Note that for all $i,j,k$ $A_i\ne Z_j$, $Z_j\ne B_k$, $A_i\ne B_k$ since $type(A_i)=1$, $type(B_k)=2$, $type(Z_j)=3$. Let $g:E_{\hat{H}}\to \Sigma$ be such a function that $g(e)=a$ if $lab_{\hat{H}}(e)=A_i$, $g(e)=b$ if $lab_{\hat{H}}(e)=B_k$, and $g(e)=z$ if $lab_{\hat{H}}(e)=Z_j$. Since $\mathrm{HL}\vdash\hat{H}\to S$, $g(\hat{H})$ belongs to $L(HGr)=L$. To complete the proof, observe that the number of $a$-labeled edges attached to a node in $g(\hat{H})$ equals the outdegree of this node in $H$, and the number of $z$-labeled edges attached to a node in $g(\hat{H})$ equals the indegree of this node in $H$; according to the definition of $L$, this number is the same for all nodes.
		
		Formally, in the above reasonings we made a one-way transition when we introduced $g$; hence, we only proved that $\widetilde{HGr}$ generates regular binary graphs. However, given a regular binary graph $H$, we can construct graphs of the form $f_H(H)$, $\hat{H}$, and $g(\hat{H})$ corresponding to it and then repeat the above reasonings.
	\end{proof}
	\begin{remark}
		Consider the language $L_3$ of all flowerbeds over $C^\prime$ of the form $\FB(MS_1, \dots, MS_n, b)$ such that $|MS_{1}|^a=n-1$. This language can also be generated by some HRG (this is left as an exercise to the reader). If we defined $L$ in the above proof as $L_1\cap L_2 \cap L_3$, then $L(\widetilde{HGr})$ would consist of all regular binary graphs with $n$ nodes such that the indegree and the outdegree of each node equals $n-1$. Note that numbers of edges in graphs of $L(\widetilde{HGr})$ in such a case form the set $\left\{\frac{n(n-1)}{2}\mid n\ge 2\right\}$, which grows with the pace $O(n^2)$ (this violates the Linear-Growth theorem, see \cite{Habel92}, Chapter IV.2).
	\end{remark}
	
	\section{Conclusion}\label{sec_conclusion}
	Hypergraph Lambek grammars are a logical formalism, which extend hyperedge replacement grammars. They deal with hypergraph types and sequents, which have a model-theoretic semantics of hypergraph languages. We showed that HL-grammars are more powerful than HRGs; in particular, they violate the pumping lemma and the Linear-Growth  theorem. Note that, since they generate the language of all graphs, they are able to generate languages of unbounded treewidth. This can be considered as a disadvantage since we cannot use algorithms for languages of bounded treewidth. However, our goal was rather to show that $\mathrm{HL}$-grammars are much more powerful than HRGs, which cannot be obtained for free.
	
	Note that the membership problem for HL-grammars is NP-complete: this follows from the fact that, if $H$ belongs to $L(HGr)$ for $HGr=\langle\Sigma,S,\triangleright\rangle$, then this can be certified by a function $f_H$ and by a derivation of $f_H(H)\to S$. Description of $f_H$ and the derivation have polynomial size with respect to $H$ and $HGr$, hence the problem is in NP. It is NP-complete since HL-grammars can generate an NP-complete language (which can be generated by some HRG without isolated nodes). Therefore, being equal in complexity to HRGs, HL-grammars represent a promising instrument for generating hypergraph languages.
	
	As is often the case, there are more questions than answers. Some of them are listed below:
	\begin{enumerate}
		\item Is it true that, if $\overline{w}(H\to A)$ is true for all valuations, then  $\mathrm{HL}\vdash H\to A$ (for sequents that include types with $\times$)?
		\item Do HL-grammars generate the language of (a) complete binary graphs? (b) planar binary graphs? (c) directed acyclic binary graphs?
		\item What string languages can be generated by HL-grammars?
		\item Is the class of languages generated by HL-grammars closed under intersections?
		\item What nontrivial necessary properties (like the pumping lemma for HRGs) exist for languages generated by HL-grammars?
		\item Can HL-grammars be embedded in some known kind of graph grammars?
	\end{enumerate}
	We are interested in further and deeper study of generalizations of logical approaches and concepts to as graphs; we think that this allows one to better understand the nature of the considered notions.
	\section*{Acknowledgments}
		I am grateful to my scientific advisor Mati Pentus for his careful attention to my studies. I also thank anonymous reviewers for their substantial and valuable advice.
\bibliographystyle{eptcs}
\bibliography{GBLHL_Post_Procs}

\begin{thebibliography}{10}
\providecommand{\bibitemdeclare}[2]{}
\providecommand{\surnamestart}{}
\providecommand{\surnameend}{}
\providecommand{\urlprefix}{Available at }
\providecommand{\url}[1]{\texttt{#1}}
\providecommand{\href}[2]{\texttt{#2}}
\providecommand{\urlalt}[2]{\href{#1}{#2}}
\providecommand{\doi}[1]{doi:\urlalt{http://dx.doi.org/#1}{#1}}
\providecommand{\bibinfo}[2]{#2}

\bibitemdeclare{inproceedings}{Bauer16}
\bibitem{Bauer16}
\bibinfo{author}{Daniel \surnamestart Bauer\surnameend} \&
  \bibinfo{author}{Owen \surnamestart Rambow\surnameend}
  (\bibinfo{year}{2016}): \emph{\bibinfo{title}{Hyperedge Replacement and
  Nonprojective Dependency Structures}}.
\newblock In \bibinfo{editor}{David \surnamestart Chiang\surnameend} \&
  \bibinfo{editor}{Alexander \surnamestart Koller\surnameend}, editors: {\sl
  \bibinfo{booktitle}{Proceedings of the 12th International Workshop on Tree
  Adjoining Grammars and Related Formalisms (TAG+12), June 29 - July 1, 2016,
  Heinrich Heine University, D{\"{u}}sseldorf, Germany}},
  \bibinfo{publisher}{The Association for Computer Linguistics}, pp.
  \bibinfo{pages}{103--111}.
\newblock \urlprefix\url{https://www.aclweb.org/anthology/W16-3311/}.

\bibitemdeclare{inproceedings}{Drewes17}
\bibitem{Drewes17}
\bibinfo{author}{Frank \surnamestart Drewes\surnameend} \&
  \bibinfo{author}{Anna \surnamestart Jonsson\surnameend}
  (\bibinfo{year}{2017}): \emph{\bibinfo{title}{Contextual Hyperedge
  Replacement Grammars for Abstract Meaning Representations}}.
\newblock In \bibinfo{editor}{Marco \surnamestart Kuhlmann\surnameend} \&
  \bibinfo{editor}{Tatjana \surnamestart Scheffler\surnameend}, editors: {\sl
  \bibinfo{booktitle}{Proceedings of the 13th International Workshop on Tree
  Adjoining Grammars and Related Formalisms, {TAG} 2017, Ume{\aa}, Sweden,
  September 4-6, 2017}}, \bibinfo{publisher}{Association for Computational
  Linguistics}, pp. \bibinfo{pages}{102--111}.
\newblock \urlprefix\url{https://www.aclweb.org/anthology/W17-6211/}.

\bibitemdeclare{incollection}{Drewes97}
\bibitem{Drewes97}
\bibinfo{author}{Frank \surnamestart Drewes\surnameend},
  \bibinfo{author}{Hans{-}J{\"{o}}rg \surnamestart Kreowski\surnameend} \&
  \bibinfo{author}{Annegret \surnamestart Habel\surnameend}
  (\bibinfo{year}{1997}): \emph{\bibinfo{title}{Hyperedge Replacement Graph
  Grammars}}.
\newblock In \bibinfo{editor}{Grzegorz \surnamestart Rozenberg\surnameend},
  editor: {\sl \bibinfo{booktitle}{Handbook of Graph Grammars and Computing by
  Graph Transformations, Volume 1: Foundations}}, \bibinfo{publisher}{World
  Scientific}, pp. \bibinfo{pages}{95--162}, \doi{10.1142/9789812384720\_0002}.

\bibitemdeclare{inproceedings}{Gilroy17}
\bibitem{Gilroy17}
\bibinfo{author}{Sorcha \surnamestart Gilroy\surnameend}, \bibinfo{author}{Adam
  \surnamestart Lopez\surnameend} \& \bibinfo{author}{Sebastian \surnamestart
  Maneth\surnameend} (\bibinfo{year}{2017}): \emph{\bibinfo{title}{Parsing
  Graphs with Regular Graph Grammars}}.
\newblock In \bibinfo{editor}{Nancy \surnamestart Ide\surnameend},
  \bibinfo{editor}{Aur{\'{e}}lie \surnamestart Herbelot\surnameend} \&
  \bibinfo{editor}{Llu{\'{\i}}s \surnamestart M{\`{a}}rquez\surnameend},
  editors: {\sl \bibinfo{booktitle}{Proceedings of the 6th Joint Conference on
  Lexical and Computational Semantics, *SEM @ACM 2017, Vancouver, Canada,
  August 3-4, 2017}}, \bibinfo{publisher}{Association for Computational
  Linguistics}, pp. \bibinfo{pages}{199--208}, \doi{10.18653/v1/S17-1024}.

\bibitemdeclare{book}{Habel92}
\bibitem{Habel92}
\bibinfo{author}{Annegret \surnamestart Habel\surnameend}
  (\bibinfo{year}{1992}): \emph{\bibinfo{title}{Hyperedge Replacement: Grammars
  and Languages}}.
\newblock {\sl \bibinfo{series}{Lecture Notes in Computer Science}}
  \bibinfo{volume}{643}, \bibinfo{publisher}{Springer},
  \doi{10.1007/BFb0013875}.

\bibitemdeclare{inproceedings}{Jones12}
\bibitem{Jones12}
\bibinfo{author}{Bevan~K. \surnamestart Jones\surnameend},
  \bibinfo{author}{Jacob \surnamestart Andreas\surnameend},
  \bibinfo{author}{Daniel \surnamestart Bauer\surnameend},
  \bibinfo{author}{Karl~Moritz \surnamestart Hermann\surnameend} \&
  \bibinfo{author}{Kevin \surnamestart Knight\surnameend}
  (\bibinfo{year}{2012}): \emph{\bibinfo{title}{Semantics-Based Machine
  Translation with Hyperedge Replacement Grammars}}.
\newblock In \bibinfo{editor}{Martin \surnamestart Kay\surnameend} \&
  \bibinfo{editor}{Christian \surnamestart Boitet\surnameend}, editors: {\sl
  \bibinfo{booktitle}{{COLING} 2012, 24th International Conference on
  Computational Linguistics, Proceedings of the Conference: Technical Papers,
  8-15 December 2012, Mumbai, India}}, \bibinfo{publisher}{Indian Institute of
  Technology Bombay}, pp. \bibinfo{pages}{1359--1376}.
\newblock \urlprefix\url{https://www.aclweb.org/anthology/C12-1083/}.

\bibitemdeclare{article}{Lambek58}
\bibitem{Lambek58}
\bibinfo{author}{Joachim \surnamestart Lambek\surnameend}
  (\bibinfo{year}{1958}): \emph{\bibinfo{title}{The Mathematics of Sentence
  Structure}}.
\newblock {\sl \bibinfo{journal}{The American Mathematical Monthly}}
  \bibinfo{volume}{65}(\bibinfo{number}{3}), pp. \bibinfo{pages}{154--170},
  \doi{10.1080/00029890.1958.11989160}.
\newblock \urlprefix\url{http://www.jstor.org/stable/2310058}.

\bibitemdeclare{inproceedings}{Pentus93}
\bibitem{Pentus93}
\bibinfo{author}{Mati \surnamestart Pentus\surnameend} (\bibinfo{year}{1993}):
  \emph{\bibinfo{title}{Lambek Grammars Are Context Free}}.
\newblock In: {\sl \bibinfo{booktitle}{Proceedings of the Eighth Annual
  Symposium on Logic in Computer Science {(LICS} '93), Montreal, Canada, June
  19-23, 1993}}, \bibinfo{publisher}{{IEEE} Computer Society}, pp.
  \bibinfo{pages}{429--433}, \doi{10.1109/LICS.1993.287565}.

\bibitemdeclare{article}{Pentus95}
\bibitem{Pentus95}
\bibinfo{author}{Mati \surnamestart Pentus\surnameend} (\bibinfo{year}{1995}):
  \emph{\bibinfo{title}{Models for the Lambek Calculus}}.
\newblock {\sl \bibinfo{journal}{Ann. Pure Appl. Log.}}
  \bibinfo{volume}{75}(\bibinfo{number}{1-2}), pp. \bibinfo{pages}{179--213},
  \doi{10.1016/0168-0072(94)00063-9}.

\bibitemdeclare{inproceedings}{Pshenitsyn20_2}
\bibitem{Pshenitsyn20_2}
\bibinfo{author}{Tikhon \surnamestart Pshenitsyn\surnameend}
  (\bibinfo{year}{2020}): \emph{\bibinfo{title}{Hypergraph Basic Categorial
  Grammars}}.
\newblock In \bibinfo{editor}{Fabio \surnamestart Gadducci\surnameend} \&
  \bibinfo{editor}{Timo \surnamestart Kehrer\surnameend}, editors: {\sl
  \bibinfo{booktitle}{Graph Transformation - 13th International Conference,
  {ICGT} 2020, Held as Part of {STAF} 2020, Online, June 25-26, 2020,
  Proceedings}}, {\sl \bibinfo{series}{Lecture Notes in Computer Science}}
  \bibinfo{volume}{12150}, \bibinfo{publisher}{Springer}, pp.
  \bibinfo{pages}{146--162}, \doi{10.1007/978-3-030-51372-6\_9}.

\bibitemdeclare{inproceedings}{Pshenitsyn20_1}
\bibitem{Pshenitsyn20_1}
\bibinfo{author}{Tikhon \surnamestart Pshenitsyn\surnameend}
  (\bibinfo{year}{2020}): \emph{\bibinfo{title}{Weak Greibach Normal Form for
  Hyperedge Replacement Grammars}}.
\newblock In \bibinfo{editor}{Berthold \surnamestart Hoffmann\surnameend} \&
  \bibinfo{editor}{Mark \surnamestart Minas\surnameend}, editors: {\sl
  \bibinfo{booktitle}{Proceedings of the Eleventh International Workshop on
  Graph Computation Models, GCM@STAF 2020, Online-Workshop, 24th June 2020}},
  {\sl \bibinfo{series}{{EPTCS}}} \bibinfo{volume}{330}, pp.
  \bibinfo{pages}{108--125}, \doi{10.4204/EPTCS.330.7}.

\bibitemdeclare{misc}{Pshenitsyn21}
\bibitem{Pshenitsyn21}
\bibinfo{author}{Tikhon \surnamestart Pshenitsyn\surnameend}
  (\bibinfo{year}{2021}): \emph{\bibinfo{title}{Introduction to a Hypergraph
  Logic Unifying Different Variants of the Lambek Calculus}}.
\newblock \urlprefix\url{https://arxiv.org/abs/2103.01199}.

\bibitemdeclare{inproceedings}{Pshenitsyn21_2}
\bibitem{Pshenitsyn21_2}
\bibinfo{author}{Tikhon \surnamestart Pshenitsyn\surnameend}
  (\bibinfo{year}{2021}): \emph{\bibinfo{title}{Powerful and NP-Complete:
  Hypergraph Lambek Grammars}}.
\newblock In \bibinfo{editor}{Fabio \surnamestart Gadducci\surnameend} \&
  \bibinfo{editor}{Timo \surnamestart Kehrer\surnameend}, editors: {\sl
  \bibinfo{booktitle}{Graph Transformation - 14th International Conference,
  {ICGT} 2021, Held as Part of {STAF} 2021, Virtual Event, June 24-25, 2021,
  Proceedings}}, {\sl \bibinfo{series}{Lecture Notes in Computer Science}}
  \bibinfo{volume}{12741}, \bibinfo{publisher}{Springer}, pp.
  \bibinfo{pages}{102--121}, \doi{10.1007/978-3-030-78946-6\_6}.

\end{thebibliography}

\end{document}